\documentclass[smallcondensed]{svjour3}
\usepackage{lineno}
\usepackage{fix-cm}
\usepackage{amsmath}
\usepackage{graphicx}
\usepackage{latexsym,amssymb,amsbsy}
\usepackage{enumitem}
\usepackage{scrextend}
\usepackage{mathtools}
\usepackage{graphics,color}
\usepackage{verbatim}
\usepackage{url}
\usepackage{pgf,tikz}
\usepackage[normalem]{ulem}
\usetikzlibrary{trees,automata,arrows,shapes.geometric}
\usetikzlibrary{calc}
\newcommand*\circled[1]{\tikz[baseline=(char.base)]{\node[shape=circle,fill=green!20,draw,inner sep=0.5pt] (char) {#1};}}
\usepackage{xfrac}
\usepackage{amsfonts}
\usepackage{algorithm}
\usepackage[noend]{algpseudocode}
\usepackage{colortbl,xcolor}
\newcommand*{\QEDB}{\hfill\ensuremath{\square}}

\newcommand{\etal}{\textit{et al. }}
\newcommand{\eg}{\textit{e.g.}}
\newcommand{\ie}{\textit{i.e.}}

\def\F{{\mathbb{F}}}

\def\bv{{\mathbf{v}}}
\def\bw{{\mathbf{w}}}
\def\bs{{\mathbf{s}}}
\def\bu{{\mathbf{u}}}

\def\0{{\mathbf{0}}}
\def\1{{\mathbf{1}}}
\def\ba{{\mathbf{a}}}
\def\bb{{\mathbf{b}}}
\def\bc{{\mathbf{c}}}
\def\bx{{\mathbf{x}}}
\def\by{{\mathbf{y}}}
\def\cG{{\mathcal{G}}}

\def\cC{{\mathcal{C}}}

\DeclareMathOperator{\nlc}{nlc}

\algdef{SE}[DOWHILE]{Do}{doWhile}{\algorithmicdo}[1]{\algorithmicwhile\ #1}

\journalname{}

\begin{document}
\title{A Graph Joining Greedy Approach to Binary de Bruijn Sequences}
\author{Zuling Chang \and\\
Martianus Frederic Ezerman \and\\
Adamas Aqsa Fahreza \and Qiang Wang}
%\authorrunning{Chang \and Ezerman \and Fahreza \and Wang}
\institute{Z. Chang \at School of Mathematics and Statistics, Zhengzhou University, Zhengzhou 450001 and\\
State Key Laboratory of Information Security, Institute of Information Engineering,\\ Chinese Academy of Sciences, Beijing 100093, China\\
\email{zuling\textunderscore chang@zzu.edu.cn}
\and
M. F. Ezerman and A. A. Fahreza \at School of Physical and Mathematical Sciences,\\
Nanyang Technological University, 21 Nanyang Link, Singapore 637371\\
\email{\{fredezerman,adamas\}@ntu.edu.sg}
\and
Q. Wang \at School of Mathematics and Statistics, Carleton University,\\
1125 Colonel By Drive, Ottawa, ON K1S 5B6, Canada\\
\email{wang@math.carleton.ca}
}
\date{Received: date / Accepted: date}
\maketitle
	
\begin{abstract}
Using greedy algorithms to generate de Bruijn sequences is a classical approach that has produced numerous interesting theoretical results. This paper investigates an algorithm which we call the {\tt Generalized Prefer-Opposite} (GPO). It includes all prior greedy algorithms, with the exception of the Fleury Algorithm applied on the de Bruijn graph, as specific instances. The GPO Algorithm can produce any binary periodic sequences with nonlinear complexity at least two on input a pair of suitable feedback function and initial state. In particular, a sufficient and necessary condition for the GPO Algorithm to generate binary de Bruijn sequences is established. This requires the use of feedback functions with a unique cycle or loop in their respective state graphs. Moreover, we discuss modifications to the GPO Algorithm to handle more families of feedback functions whose state graphs have multiple cycles or loops. These culminate in a graph joining method. Several large classes of feedback functions are subsequently used to illustrate how the GPO Algorithm and its modification into the {\tt Graph Joining Prefer-Opposite} (GJPO) Algorithm work in practice.
		
% \PACS{PACS code1 \and PACS code2 \and more}
\end{abstract}
%%%%
	
\section{Introduction}\label{sec:intro}
	
Let $n$ be any positive integer. In a binary {\it de Bruijn sequence} of order $n$, within any string of length $2^n$, each $n$-tuple occurs exactly once. There are $2^{2^{n-1}-n}$ cyclically inequivalent de Bruijn sequences~\cite{Bruijn46}. While many properties of such sequences are well-known, still a lot more remains to be discovered as applications in diverse areas such as cryptography, bioinfomatics, and robotics continue to be discussed. Instead of attempting a comprehensive survey of the application landscape, we highlight a number of implementations in the literature.
	
It is a major interest in cryptography to identify a large number of highly nonlinear de Bruijn sequences and, if possible, to quickly generate them, either on hardware or software. Works on the design and evaluation of de Bruijn sequences for cryptographic implementations are too numerous to list. We mention a few as starting points for interested readers to consult. Some de Bruijn sequences with certain linear complexity profiles (see the discussion in~\cite{Stamp1993}) can be used as keystream generators in stream chipers~\cite{Robshaw2008}. An efficient hardware implementation, capable of handling large order $n$, was proposed in~\cite{Yang2017}. 
	
An influential work of Pevzner \etal showed deep connections between Eulerian paths in de Bruijn graphs and the fragment assembly problem in DNA sequencing~\cite{PTW01}. More recent works, \eg, that of Aguirre \etal in~\cite{Aguirre2011} described how to use de Bruijn sequences in neuroscientific studies on the neural response to stimuli. The sequences form a rich source of hard-to-predict orderings of the stimuli in the experimental designs.
	
Position locators are extremely useful in robotics and often contain de Bruijn sequences as ingredients. Scheinerman determined the absolute $2$-dimensional location of a robot in a workspace by tiling it with black and white squares based on a variant of a chosen de Bruijn sequence in~\cite{Scheinerman2001}. The roles of de Bruijn sequences in robust positioning patterns, with some discussion on their error-control capabilities, are treated in~\cite{Bruck12}. A nice application of de Bruijn sequences in image acquisition was proposed in~\cite{Pages2005}. They form a design of coloured stripe patterns to locate both the intensity peaks and the edges without loss of accuracy, while reducing the number of required hue levels. A variant of the $32$-card magic trick, based on a de Bruijn sequence of order $5$, inspired Gagie to come up with two new lower bounds for data compression and a new lower bound on the entropy of a Markov source in~\cite{Gagie2012}.
	
Extensive studies on a topic of long history such as de Bruijn sequences must have resulted in numerous generating methods. An excellent survey for approaches up to the end of 1970s can be found in~\cite{Fred82}. Numerous construction routes continue to be proposed. One can arguably put them into several clusters. 
	
In a graph-theoretic construction, one may start with the $n$-dimensional de Bruijn graph over the binary symbols. A de Bruijn sequence of order $n$ is a Hamiltonian path in the graph. Equivalently, the sequence is an Eulerian cycle in the $(n-1)$-dimensional de Bruijn graph. A complete enumeration can be inefficiently done, \eg, by the Fleury Algorithm~\cite{Fleury}. Scores of ideas on how to identify specific paths or cycles have also been put forward, typically with additonal properties imposed on the sequences.
	
Some approaches are algebraic. One can take a primitive polynomial of order $n$ in the ring of polynomials $\F_2[x]$, turn it into a linear Feedback Shift Register (FSR), and output a maximal length sequence, also known as $m$-sequence. Appending another $0$ to the string of $0$s of length $n-1$ in the $m$-sequence results in a de Bruijn sequence. The {\it cycle joining  method} (CJM) is another well-known generic construction approach. As discussed in, \eg, \cite{Fred82} and~\cite{Golomb}, the main idea is to join all cycles produced by a given FSR into a single cycle by identifying the so-called conjugate pairs. 
	
Numerous FSRs have been shown to yield a large number of output sequences. We mention two of the many works available in the literature. Etzion and Lempel chose the {\it pure cycling register} (PCR) and the {\it pure summing register} (PSR) in~\cite{EL84} to generate a remarkable number, exponential in $n$, of de Bruijn sequences. When the characteristic polynomials of the linear FSRs are product of pairwise distinct irreducible polynomials, Chang \etal proposed some algorithms to determine the exact number of de Bruijn sequences that the CJM can produce in~\cite{Chang2019}. Interested readers may want to consult~\cite[Table 4]{Chang2019} for a summary of the input parameters and the performance complexity of prior works. There have also been a lot of constructions based on {\it ad hoc} rules. They are cheap to implement but yield only a few de Bruijn sequences. An example is the joining, in lexicographic order, of all Lyndon words whose length divides $n$ in~\cite{FM78}. Jansen \etal established a requirement to determine some conjugate pairs in~\cite{JFB91}, leading to an efficient generating algorithm. A special case of the requirement was highlighted in~\cite{Sawada16} and a generalization was subsequently given by Gabric \etal in~\cite{Gabric2018}. 
	
In the {\it cross-join pairing} method, one starts with a known de Bruijn sequence and proceeds to identify cross-join pairs that allow the sequence to be cut and reglued into inequivalent de Bruijn sequences. Further details and examples were supplied, \eg, by Helleseth and Kl{\o}ve in~\cite{HK91} and by Mykkeltveit and Szmidt in~\cite{MS15}.
	
Our present work focuses on the greedy algorithms, of which {\tt Prefer-One}~\cite{Martin1934} is perhaps the most famous, followed by {\tt Prefer-Same}~\cite{Fred82} and {\tt Prefer-Opposite}~\cite{Alh10}. Although using greedy algorithms to generate de Bruijn sequences tends to be impractical due to the usually exponential storage demand, it remains theoretically interesting. Aside from being among the oldest methods of generating de Bruijns sequences, greedy algorithms often shed light on various properties of de Bruijn sequences and their connections to other combinatorial and discrete structures. New greedy algorithms continue to be discussed, \eg, a recent one in~\cite{WWZ18}. The current first three authors recently proposed the {\tt Generalized Prefer-Opposite} (GPO) Algorithm to generalize known greedy algorithms and provided several new families of de Bruijn sequences in~\cite{Chang21}. A good number of greedy algorithms make use of {\it preference functions} as important construction tools. A discussion on this approach was given in~\cite{Golomb}. Many results on preference functions have been established by Alhakim in~\cite{Alh12}.
	
Here we prove further properties of the GPO Algorithm and show that this new approach covers \emph{all} previously-known ad hoc greedy algorithms as well as those based on preference functions. More specifically, our investigation yields the following contributions.
	
\begin{enumerate}
	\item The first contribution is summarized in Theorems~\ref{gpo} and~\ref{thm:m}. 
		
	The fact that \emph{any} periodic sequence can be produced by a greedy algorithm is well-known. Given the sequence, however, it is often unclear which preference or feedback function and initial state yield the sequence. Theorem~\ref{gpo} provides an answer. It confirms that \emph{any} periodic sequence with nonlinear complexity $n \geq 2$ can be produced by the GPO Algorithm. It then shows how to \emph{explicitly find} an input pair. The pair consists of a suitable initial state and a feedback function $f(x_0,x_1,\ldots,x_{n-1})$, in which the coefficient of $x_0$ is zero. Theorem~\ref{gpo} also corrects a minor mistake~\cite[Lemma 3 on p. 132]{Golomb} in the influential book of Golomb.
		
	Theorem~\ref{thm:m} generalizes some well-known results from \cite{Fred82} and \cite{Alh12}. Since all de Bruijn sequences can be generated by some greedy algorithms with the all zeroes $0,0,\ldots,0$ initial state, prior works in the literature only studied the case of that initial state. Theorem~\ref{thm:m}, on the other hand, treats the general case that \emph{the initial state can be chosen more arbitrarily}. In particular, the GPO Algorithm generates a de Bruijn sequence if and only if the initial state is chosen from the states contained in the only cycle in the state graph of the FSR of a suitably chosen feedback function.
		
	\item Our second contribution is arguably much more significant. Sections~\ref{sec:GJ} and~\ref{sec:PAG} discuss a novel {\bf graph joining algorithm} in detail. It generalizes the cycle joining approach by removing the requirement that the FSR under consideration only produces disjoint cycles. More specifically, if a particular function $f$ fails to directly generate any de Bruijn sequence via the GPO Algorithm but meets a simple condition (details will be given later), then we modify the algorithm, by adding assignment rules, to ensure that the output sequences become de Bruijn. This graph joining method (GJM) results in our {\tt Graph Joining Prefer-Opposite} (GJPO) Algorithm.
\end{enumerate}
	
We note that the results extend naturally to nonbinary setups. A more comprehensive treatment, however, lies beyond the scope of our present construction work. So is a deeper analysis on the linear complexity profiles of the resulting sequences. 
	
After this introduction, Section~\ref{sec:prelim} gathers useful preliminary notions and known results. Section~\ref{sec:periodic} establishes the conditions for the GPO Algorithm to generate periodic sequences. Section~\ref{sec:suffice} characterizes the conditions for which the GPO Algorithm produces de Bruijn sequences. Section~\ref{sec:GJ} shows how to modify the GPO Algorithm to still produce de Bruijn sequences even when the required conditions are not satisfied. This allows us to include larger classes of feedback functions in our new graph-joining construction method. We study the graph theoretic properties of a class of feedback functions whose structures can be easily determined and succinctly stored in Section~\ref{sec:PAG}. This shows that mitigation is possible to reduce the steep storage cost as $n$ grows, if choices are made judiciously. The last section contains a summary and a few directions to consider. Observations to highlight how the GJM differs from the CJM wrap the paper up.
	
\section{Preliminaries}\label{sec:prelim}
	
An {\it $n$-stage shift register} is a circuit. It has $n$ consecutive storage units and is clock-regulated. Each unit holds a bit. As the clock pulses, the bit shifts to the next stage in line. The register outputs a new bit $s_n$ based on the $n$-bit {\it initial state} $\bs_0 := s_0,\ldots,s_{n-1}$. The corresponding {\it feedback function} $f(x_0,\ldots,x_{n-1})$ is the Boolean function over $\F_2^{n}$ that outputs $s_n$ on input $\bs_0$.
	
The output of a feedback shift register (FSR) is a binary sequence $\bs := s_0,s_1,\ldots, s_n$, $\ldots$ satisfying $s_{n+\ell} = f(s_{\ell},s_{\ell+1}, \ldots, s_{\ell+n-1})$ for $\ell = 0,1,2,\ldots$. The \emph{smallest integer} $N$ that satisfies $s_{i+N} = s_i$ for all $i \geq 0$ is the {\it period} of $\bs$. The {\it $N$-periodic} sequence $\bs$ can then be written as $\bs:= (s_0,s_1,s_2,\ldots,s_{N-1})$. We call $\bs_i := s_i,s_{i+1},\ldots,s_{i+n-1}$ the 	{\it $i$-th state}\footnote{In most references, a state is written in between parentheses, for example, $\bs_i :=(s_i,s_{i+1},\ldots,s_{i+n-1})$. We remove the parentheses, throughout, for brevity and to avoid the cumbersome $f((s_i,s_{i+1},\ldots,s_{i+n-1}))$ notation for state evaluation.} of $\bs$. The states $\bs_{i-1}$ and $\bs_{i+1}$, analogously defined, are the {\it predecessor} and {\it successor} of $\bs_i$, respectively. In tabular form, an $n$-string $c_0,c_1,\ldots,c_{n-1}$ is often written concisely as $c_0c_1 \ldots c_{n-1}$. The {\it complement} $\overline{c}$ of $c \in \F_2$ is $1+c$. The complement of a string is produced by taking the complement of each element. The respective strings of zeroes and of ones, each of length $\ell$, are denoted by $\0^{\ell}$ and $\1^{\ell}$. Given an $n$-stage state $\ba :=a_0,a_1,\ldots,a_{n-1}$, its {\it conjugate state} is $\widehat{\ba}:=\overline{a_0}, a_1, \ldots, a_{n-1}$ while $\widetilde{\ba}:=a_0,a_1,\ldots,\overline{a_{n-1}}$ is its {\it companion state}.
	
A feedback function $f(x_0,x_1,\ldots,x_{n-1}) = x_0 + g(x_1,\ldots,x_{n-1})$, where $g$ is a Boolean function over $\mathbb{F}_2^{n-1}$, is said to be {\it nonsingular}. Otherwise $f$ is said to be {\it singular}. An FSR with nonsingular feedback function generates periodic sequences~\cite[pp. 115--116]{Golomb}.
	
The FSR that generates a binary (ultimately) periodic sequence $\bs$ is, in general, not unique. In his doctoral thesis~\cite{Jansen89}, Jansen introduced the notion of {\it nonlinear complexity}. The nonlinear complexity of $\bs$, denoted by $\nlc(\bs)$, is the smallest $k$ such that there is a $k$-stage FSR, with feedback function $f(x_0, x_1, \ldots, x_{k-1})$, that generates $\bs$. When this holds, we can make three assertions. 
\begin{enumerate}
	\item Each $k$-stage state of $\bs$ appears at most once. 
	\item There exists at least one $(k-1)$-stage state of $\bs$ that appears twice. 
	\item The period of $\bs$ is at most $2^k$.
\end{enumerate}

Seen in this light, binary de Bruijn sequences of order $n$ are periodic sequences with nonlinear complexity $n$ having the maximum period. Furthermore, there are only two periodic sequences with nonlinear complexity $0$, namely, the all zero sequence $(0,0,\ldots,0)$ with feedback function $f=0$ and the all one sequence $(1,1,\ldots,1)$ with $f=1$. The unique sequence with nonlinear complexity $1$ is $(0,1)$ with $f=1+x_0$. Henceforth, we concentrate on a binary periodic sequence $\bs$ with $\nlc(\bs) \geq 2$.
	
The {\it state graph} of FSR with feedback function $f(x_0,x_1,\ldots,x_{n-1})$ is a directed graph $\cG_f$. Its vertex set $V_{\cG_f}$ consists of all $2^n$ $n$-stage states. There is an edge directed from $\bu:=u_0,u_1,\ldots,u_{n-1}$ to $\bv:=u_1,u_2,\ldots,f(u_0,u_1,\ldots,u_{n-1})$. We call $\bu$ a {\it child} of $\bv$ and $\bv$ the {\it parent} of $\bu$. We allow $\bu = \bv$, in which case $\cG_f$ contains a loop, \ie, a cycle with only one vertex. Notice that if $f$ is singular, then some states have two distinct children in $\cG_f$. A {\it leaf} in $\cG_f$ is a vertex with no child. A vertex $\bw$ is a {\it descendent} of $\bu$ if there is a directed path that starts at $\bw$ and ends at $\bu$. In turn, $\bu$ is called an {\it ancestor} of $\bw$. A {\it rooted tree} $T_{f,\bb}$ in $\cG_f$ is the largest tree in $\cG_f$ in which the vertex $\bb$ has been designated the {\it root} and the edge that emanates out of $\bb$ has been removed from $\cG_f$. In this work the orientation is \textbf{towards} the root $\bb$, \ie, $T_{f,\bb}$ is an {\it in-tree} or an {\it anti-arborescence}.
	
{\it Preference function} is an important tool in constructing $t$-ary de Bruijn sequences by greedy algorithm \cite[Chapter VI]{Golomb}. We first recall the next two definitions.
	
\begin{definition}\label{def:pf}
Given a positive integer $t>1$ and an alphabet $A =\{0, 1, \ldots, t-1\}$ of size $t$, a preference function $P$ of $n-1$ variables is a $t$-dimensional vector valued function of $(n-1)$-stage states such that $P_0(\ba),\ldots, P_{t-1}(\ba)$ is a rearrangement of $0,1,\ldots,t-1$, for each choice of $\ba=a_0, a_1, \ldots, a_{n-2} \in A^{n-1}$.
\end{definition}
	
\begin{definition}\label{def:Galg}
For any initial state $\bu=u_0,u_1,\ldots,u_{n-1}$ and preference function $P$, the following inductive definition determines a unique finite sequence $\bs$.
\begin{enumerate}
	\item $s_0=u_0,\, s_1=u_1,\,\ldots,\, s_{n-1}=u_{n-1}$.
	\item If $s_{N+1},\ldots,s_{N+n-1}$ have been defined, then 
	\[
	s_{N+n} := P_i(s_{N+1},\ldots,s_{N+n-1}),
	\]
	where $i$ is the least integer such that the state 
	\[
	s_{N+1},\ldots,s_{N+n-1},P_i(s_{N+1}, \ldots,s_{N+n-1})
	\]
	has not previously occurred in $s_0,s_1,\ldots,s_{N+n-1}$.
	\item Let $L(s_i)$ be the first value of $N$ such that no integer $i$ can be found to satisfy Item 2. Then the last digit of the sequence is $s_{L(s_i)+n-2}$ and $L(s_i)$ is called the cycle period.
\end{enumerate}
\end{definition}
	
The sequence generated by the algorithm described in Definition~\ref{def:Galg} is known to be $L(s_i)$-periodic. Golomb presented various ways to use preference function to generate de Bruijn sequences in~\cite[Chapter VI]{Golomb}. In the next section we modify the procedure in Definition \ref{def:Galg} to design a new algorithm that generates periodic sequences.
	
\section{GPO Algorithm and Periodic Sequences}\label{sec:periodic}
	
	When $t=2$, a Boolean function $f(x_0,x_1,\ldots,x_{n-1})$ over $\mathbb{F}_2^n$ that satisfies
	\begin{equation}\label{eq:one}
		f(0,\ba)=f(1,\ba)=P_1(\ba)
	\end{equation}
	is equivalent to a preference function. The procedure in Definition~\ref{def:Galg} can be carried out by Algorithm~\ref{algo:genpo}. Following~\cite{Chang21}, we called it the {\tt Generalized Prefer-Opposite} (GPO) Algorithm. The {\tt Prefer-One} de Bruijn sequence, for example, is the output on input $f=0$ and $\bu=\0^n$.
	
	\begin{algorithm}[ht!]
		\caption{{\tt Generalized Prefer-Opposite} (GPO)}
		\label{algo:genpo}
		\begin{algorithmic}[1]
			\renewcommand{\algorithmicrequire}{\textbf{Input:}}
			\renewcommand{\algorithmicensure}{\textbf{Output:}}
			\Require A feedback function $f(x_0,x_1,\ldots,x_{n-1})$ and an initial state $\bu$.
			\Ensure A binary sequence.
			\State{$\bc:=c_0, c_1, \ldots ,c_{n-1} \gets \bu$}	
			\Do
			\State{Print($c_0$)}
			\State{$y \gets f(c_0,c_1,\ldots,c_{n-1})$}
			\If{$c_1, c_2, \ldots, c_{n-1}, \overline{y}$ has not appeared before}
			\State{$\bc \gets c_1, c_2, \ldots, c_{n-1},  \overline{y}$} \label{Algpo:line5}
			\Else
			\State{$\bc \gets c_1, c_2, \ldots, c_{n-1}, y $} \label{Algpo:line7}
			\EndIf
			\doWhile{$\bc \neq \bu$}
		\end{algorithmic}
	\end{algorithm}
	
	\begin{remark}
		We make some remarks on the GPO Algorithm in relation to preference functions and the procedure in Definition~\ref{def:Galg}.
		
		\begin{enumerate}
			\item If the current state is $c_0, c_1, \ldots ,c_{n-1}$, then the GPO Algorithms prefers 
			\[
			\bc:=c_1, c_2, \ldots, c_{n-1}, \overline{f(c_0,\ldots,c_{n-1})}
			\]
			as the next state, unless $\bc$ had appeared before. The algorithm prefers the opposite of $f(c_0,\ldots,c_{n-1})$. The name {\tt Prefer-Opposite} Algorithm had already been used in~\cite{Alh10}, so we call Algorithm~\ref{algo:genpo} the {\tt Generalized Prefer-Opposite} (GPO) Algorithm.
			\item The GPO Algorithm and the procedure in Definition~\ref{def:Galg} differ slightly. If the initial state is $\bu=u_0,u_1,\ldots,u_{n-1}$, then $P_0(u_0,\ldots,u_{n-2}):=u_{n-1}$ in Definition~\ref{def:Galg}. Instead of $P_1(u_0,\ldots,u_{n-2}):=\overline{u_{n-1}}$, we let $f(x,u_0,\ldots,u_{n-2}):=u_{n-1}$, for any $x \in \mathbb{F}_2$, to ensure that the GPO Algorithm terminates.
			\item We restrict our attention to the binary case because of its wider interest, even though the GPO Algorithm extends naturally to the nonbinary case.
		\end{enumerate}
	\end{remark}
	
	To successfully use the GPO Algorithm to generate binary de Bruijn sequences, suitable feedback functions and initial states must be identified. For this purpose, we begin by studying the properties of this algorithm to establish the condition that ensures the output is de Bruijn.
	
	When does the GPO Algorithm generate a periodic sequence? If the output is periodic, is it necessarily de Bruijn? Do distinct input pairs always yield inequivalent outputs? The answer to all three questions is no. For some input pairs $(f,\bu)$, the algorithm never revisits the initial state $\bu$ and, hence, does not terminate. In such cases, the output sequence is {\it ultimately periodic} but not periodic. On occasions, the algorithm generates cyclically equivalent sequences on distinct input pairs $\{(f_i,\bu_i): i \in I \}$ for some index set $I$. Take $f_1(x_0,\ldots,x_{n-1}):=0$ and $f_2(x_0,\ldots,x_{n-1}):=\prod_{j=0}^{n-1} x_j$, with $\bu=\0^n$, for example. Both input pairs $(f_1,\bu)$ and $(f_2,\bu)$ yield an identical de Bruijn sequence of order $n$.
	
	Next we show that any binary periodic sequence with $\nlc(\bs) > 1$ can be generated by the GPO algorithm, with well-chosen feedback function and initial state.
	
	\begin{theorem}\label{gpo}
		Let $\bs$ be any binary periodic sequence with $\nlc(\bs) = n \ge 2$. Then $\bs$ can be generated by the GPO Algorithm with input pair $(f,\bu)$ provided that the following conditions are met.
		\begin{enumerate}[noitemsep]
			\item The feedback function $f$ has the form
			\begin{equation}\label{equ:f}
				f(x_0,x_1,\ldots,x_{n-1}):=g(x_1,\ldots,x_{n-1})
			\end{equation}
			for some Boolean function $g(x_0,\ldots,x_{n-2})$ over $\F_2^{n-1}$.
			\item The state $\bu=u_0,\ldots,u_{n-2},u_{n-1}$ meets the requirement that $u_0,\ldots,u_{n-2}$ appears twice in each period of $\bs$.
		\end{enumerate}
	\end{theorem}
	
	\begin{proof}
		Let an $N$-periodic sequence $\bs=(s_0,s_1,\ldots,s_{N-1})$ be given. In one period of $\bs$, since $\nlc(\bs)=n$, there must exist at least one $(n-1)$-stage state $a_0,a_1,\ldots,a_{n-2}$ that appears twice. Hence, the $n$-stage state $\ba=a_0,\ldots,a_{n-2},a$ and its companion state $\widetilde{\ba}=a_0,\ldots,a_{n-2},\overline{a}$ with $a \in \F_2$ appear once each. Any of these two can be the initial state of the GPO Algorithm. Without loss of generality, let $\bu := \bs_0 = s_0,s_1,\ldots,s_{n-1} = \ba$.
		
		We construct a feedback function $f(x_0,\ldots, x_{n-1})$ based on $\bs$. Let
		\[
		f(x_0,a_0,\ldots,a_{n-2}):=g(a_0,\ldots,a_{n-2})=a \mbox{ for } x_0 \in \F_2.
		\]
		We start with the initial state $\bu$. For any state $\bs_i = s_i,s_{i+1},\ldots,s_{i+n-1}$ with $0 \leq i \leq N-2$, if the $(n-1)$-stage state $s_{i+1},\ldots,s_{i+n-1} \neq a_0,\ldots,a_{n-2}$ appears for the first time in $\bs$, then we define
		\[
		f(s_i,s_{i+1},\ldots,s_{i+n-1}) :=f(\overline{s_i},s_{i+1},\ldots,s_{i+n-1}) = g(s_{i+1},\ldots,s_{i+n-1})= \overline{s_{i+n}}. %= s_{i+n}+1
		\]
		If $s_{i+1},\ldots,s_{i+n-1}$ appears for the second time in $\bs$, then we define
		\[
		f(s_i,s_{i+1},\ldots,s_{i+n-1}) := f(\overline{s_i},s_{i+1},\ldots,s_{i+n-1}) = g(s_{i+1},\ldots,s_{i+n-1})=s_{i+n}.
		\]
		The above two definitions of $f$ coincide if the $(n-1)$-stage state $s_{i+1},\ldots,s_{i+n-1}$ appears twice in $\bs$. If the $(n-1)$-stage state $b_{1},\ldots,b_{n-1}$ has not appeared in $\bs$, we define
		\[
		f(x_0,b_1,\ldots,b_{n-1}) := g(b_1,\ldots,b_{n-1})=b \mbox{ for any } x_0
		\mbox{ and } b \in \F_2,
		\]
		that is, $g(b_1,\ldots,b_{n-1})$ can take any arbitrary binary value.
		
		Now we prove that the output of the GPO Algorithm on input $(f,\bu)$ defined above is indeed $\bs$. As the run of the algorithm begins, $\bc=c_0,\ldots,c_{n-1} = \bu =s_0,\ldots,s_{n-1}$. Inductively, suppose that $c_i,\ldots,c_{i+n-1}=s_i,\ldots,s_{i+n-1}$ for some $i$ with $0\leq i\leq N-2$ and in sequence $\bs$ the bit after the state $s_i,\ldots,s_{i+n-1}$ is $s_{i+n}$.
		If $s_{i+1},\ldots,s_{i+n-1}=a_0,\ldots,a_{n-2}$, then this $(n-1)$-stage state appears in $\bs$ for the second time and $s_{i+n}=\overline{a}$. Because the state
		\begin{multline*}
			\bv=c_{i+1},\ldots,c_{i+n-1},\overline{f(c_i,c_{i+1},\ldots,c_{i+n-1})}\\
			= a_0,\ldots,a_{n-2},\overline{f(c_i,a_0,\ldots,,a_{n-2})}=a_0,\ldots,a_{n-2},\overline{a}
		\end{multline*}
		has not appeared, the algorithm rules the next state to be $s_{i+1},\ldots,s_{i+n-1},\overline{a}$ and $c_{i+n}=\overline{a}= s_{i+n}=$.
		
		Let $s_{i+1},\ldots,s_{i+n-1}\neq a_0,\ldots,a_{n-2}$. We consider the state
		\begin{multline*}
			\bv=c_{i+1},\ldots,c_{i+n-1},\overline{f(c_i,c_{i+1},\ldots,c_{i+n-1})}\\ 
			= s_{i+1},\ldots,s_{i+n-1},\overline{f(s_i,s_{i+1},\ldots,s_{i+n-1})}.
		\end{multline*}
		If $s_{i+1},\ldots,s_{i+n-1}$ appears for the first time, then, by the definition of $f$, we have 
		\[
		f(s_i,s_{i+1},\ldots,s_{i+n-1})=\overline{s_{i+n}}\]
		and \[\bv=s_{i+1},\ldots,s_{i+n-1},\overline{\overline{s_{i+n}}}=s_{i+1},\ldots,s_{i+n-1},s_{i+n}
		\]
		has not appeared. The algorithm then dictates the next state to be $s_{i+1},\ldots,s_{i+n}$ and $c_{i+n}=s_{i+n}$. If $s_{i+1},\ldots,s_{i+n-1}$ appears for the second time, then, by the definition of $f$, we know that $f(s_i,s_{i+1},\ldots,s_{i+n-1})={s_{i+n}}$ and
		\[
		\bv=s_{i+1},\ldots,s_{i+n-1},\overline{s_{i+n}}
		\]
		must have appeared earlier. Hence, the next state must be
		\[
		s_{i+1},\ldots,s_{i+n-1},f(s_i,s_{i+1}, \ldots,s_{i+n-1}) = s_{i+1},\ldots,s_{i+n-1},s_{i+n}
		\]
		and $c_{i+n}=s_{i+n}$.
		
		Finally, for $i=N-1$, the state after
		\[
		c_{N-1},c_{N},\ldots,c_{N+n-2}=s_{N-1},s_0,\ldots,s_{n-2}=s_{N-1},a_0,\ldots,a_{n-2}
		\]
		must be $a_0,\ldots,a_{n-2},a=\bu$, since $a_0,\ldots,a_{n-2},\overline{a}$ has appeared earlier. The algorithm terminates and outputs $\bs$. \qed
	\end{proof}
	
	Theorem~\ref{gpo} states that any periodic sequence with nonlinear complexity $\geq2$ can be GPO-generated with feedback function having $0$ as the coefficient of $x_0$ and some initial state. 
	The proof of Theorem \ref{gpo}, furthermore, tells us how to choose the feedback function and initial state in the GPO Algorithm to generate the periodic sequence.
	Henceforth, we just consider such feedback functions and we say that a feedback function $f$ is {\it standard} or {\it in the standard form} if it has the form specified in Equation (\ref{equ:f}). For such an $f$, each non-leaf vertex in $\cG_f$ has two children and, if there is a loop, one of the vertex's two children is the vertex itself.
	
	The following lemma puts a necessary and sufficient condition on the initial state for the GPO Algorithm to generate a periodic sequence.
	
	\begin{lemma}\label{lemma-is}
		Let a standard $f(x_0,x_1,\ldots,x_{n-1})$ and a state $\bu$ be given. Then the GPO Algorithm on $(f,\bu)$ generates a periodic sequence if and only if $\bu$ is not a leaf in $\cG_f$.
	\end{lemma}
	
	\begin{proof}
		Since $f$ is in the standard form, $\cG_f$ contains leaves. Suppose that the initial state $\bu$ is a leaf. To have
		\[
		\ba=a_0,a_1,\ldots,a_{n-1} \mbox{ and } \widehat{\ba}=\overline{a_0},a_1,\ldots,a_{n-1}
		\]
		as its two possible predecessors, $\bu$ must be
		\[a_1,\ldots,a_{n-1},\overline{f(a_0,a_1,\ldots,a_{n-1})}=a_1,\ldots,a_{n-1},\overline{g(a_1,\ldots,a_{n-1})}.
		\]
		As the algorithm visits either $\ba$ or $\widehat{\ba}$, we keep in mind that 
		\[
		a_1,\ldots,a_{n-1}, \overline{f(a_0,a_1,\ldots,a_{n-1})} = \bu
		\]
		had appeared before, as it is the initial state. Hence, the next state must be
		\[
		a_1,\ldots,a_{n-1},{f(a_0,a_1,\ldots,a_{n-1})}=a_1,\ldots,a_{n-1},{g(a_1,\ldots,a_{n-1})}=\widetilde{\bu}.
		\]
		This implies that the initial state $\bu$ will never be revisited. Thus,
		the output is an infinite and ultimately periodic sequence, \ie, not a periodic one.
		
		Now, let $\bu$ be a non-leaf state in $\cG_f$. Since the number of states in $\cG_f$ is finite, the algorithm must eventually revisit some state. We show that it is impossible for the algorithm to visit any state twice before it visits the initial state $\bu$ for the second time. For a contradiction, suppose that $\ba \neq \bu$ is the first state to be visited twice. The above analysis confirms that $\ba$ is not a leaf and, hence, has two children in $\cG_f$, say $\bv$ and $\widehat{\bv}$. When $\ba$ is visited by the algorithm for the first time, one of its two children, say $\bv$, is the actual predecessor of $\ba$ as the algorithm runs. This implies that $\widehat{\bv}$ has been visited before $\bv$. Otherwise, the successor of $\bv$ must be $\widetilde{\ba}$, by the rules of the GPO Algorithm. So we have deduced that both children of $\ba$ must have been visited by the time $\ba$ is visited for the first time. The second time $\ba$ is visited, one of its two children must have also been visited twice. This contradicts the assumption that $\ba$ is the first vertex to have been visited twice.
		
		We have thus shown that the initial state $\bu$ must be the first state to be visited twice. The algorithm stops and outputs a periodic sequence when it reaches $\bu$ for the second time. \qed
	\end{proof}
	
	The next corollary follows immediately from Theorem~\ref{gpo}.
	
	\begin{corollary}\label{cor-3e}
		Let a given binary periodic sequence $\bs$ with $\nlc(\bs)=n \ge 2$ be generated by the GPO Algorithm on input $(f,\bu)$, where $f(x_0,x_1,\ldots,x_{n-1})$ is in the standard form and $\bu=u_0,u_1,\ldots,u_{n-1}$. Then the $(n-1)$-stage state $u_0,\ldots,u_{n-2}$ appears twice in each period of $\bs$.
	\end{corollary}
	
	We note that if a periodic sequence $\bs$ is generated by the GPO Algorithm with a standard feedback function $f(x_0,\ldots,x_{n-1})$, then any $n$-stage state appears in $\bs$ at most once.
	
	\begin{corollary}\label{cor-3}
		On input a standard feedback function $f(x_0,\ldots,x_{n-1})$, the nonlinear complexity of the periodic sequence generated by the GPO Algorithm is $n$.
	\end{corollary}
	
	\begin{proof}
		Since all of the $n$-stage states of the generated sequence $\bs$ are distinct, the nonlinear complexity satisfies $\nlc(\bs)\leq n$. On the other hand, by Corollary \ref{cor-3e}, there is at least one $(n-1)$-stage state in $\bs$ that appears twice, making $\nlc(\bs)\geq n$. \qed
	\end{proof}
	
	When the GPO Algorithm generates periodic sequence, a known result in~\cite{Alh12}, also stated as~\cite[Lemma 1]{Chang21} with an alternative proof, follows from the proof of Lemma \ref{lemma-is}.
	
	\begin{lemma}\label{lemma1}
		Let $\cG_f$ be the state graph of the FSR with standard feedback function $f$. Let $\bv$ be any vertex with two children. By the time the GPO Algorithm, on input $(f,\bu)$ where $\bu$ is not a leaf, visits $\bv \neq \bu$, it must have visited both children of $\bv$.
	\end{lemma}
	
	Lemma \ref{lemma1} does not hold if the initial state is a leaf. Let $f(x_0,x_1,x_2)=x_1+1$, for example. If we use $010$, which is a leaf in $\cG_f$, as the initial state, then the states appear in the order
	\[
	010 \rightarrow 101 \rightarrow 011 \rightarrow 111 \rightarrow \dots.
	\]
	When $011$ is visited, one of its children, $001$, has not been visited.
	
	We discover that~\cite[Lemma 3 on p. 132]{Golomb} is in fact incorrect. The original statement is reproduced here for convenience.
	\smallskip
	\begin{quote}
		``For any cyclic recursive sequence \{$a_i$\} of degree $n$ and for any $n$-digit word in it, there exists a preference function which generates the sequence, using the given word as the initial word.''
	\end{quote}
	\smallskip
	The quoted statement holds for $t$-ary de Bruijn sequences, but fails in general. In the binary case, Corollary~\ref{cor-3e} tells us that the initial state $\bu$ must have the property that the $(n-1)$-stage state $u_0,\ldots,u_{n-2}$ appears twice in $\bs$. For a general periodic sequence $\bs$, such a condition may not be met. Consider $\bs=(0011101)$, for instance, and start with $001$ as the initial state. Then we have either $(P_0(00)=0 \wedge P_1(00)=1)$ or $(P_0(00)=1 \wedge P_1(00)=0)$.
	If the former holds, then, upon the second visit to $00$, the smallest $i$ is therefore $0$ and the string $000$ is formed. If the latter holds, then, upon the second visit to $00$, the smallest $i$ is now $1$ and, again, the string $000$ is formed. But $\bs$ does not contain the string $000$.
	
	Furthermore, in the $t$-ary case, the initial state must satisfy that $u_0,\ldots,u_{n-2}$ appears $t$ times in $\bs$. Thus, there exist periodic sequences of degree $n$ that can never be generated by the algorithm described in Definition~\ref{def:Galg}. We state the conclusion as a lemma, without a proof.
	
	\begin{lemma}\label{lemma-Golomb}
		For a cyclic recursive sequence \{$a_i$\} of degree $n$ such that there is an $(n-1)$-digit $u_0,\ldots,u_{n-2}$ that appears $t$ times in a period of the sequence, there exists a preference function which generates the sequence, using an initial state whose first $n-1$ digits are $u_0,\ldots,u_{n-2}$.
	\end{lemma}
	
	We have, thus far, gained two insights. 
	
	First, if the nonlinear complexity of a given periodic sequence is $n \geq 2$, then the order of the corresponding feedback function (or preference function) must also be $n$ and the chosen initial state ${\bf u}=u_0,\ldots,u_{n-2},u_{n-1}$ must satisfy the requirement that $u_0,\ldots,u_{n-2}$ appears twice in the given sequence. 
	
	Second, if either the nonlinear complexity is not equal to the order of feedback function or $u_0,\ldots,u_{n-2}$ appears just once, then the given sequence cannot be generated by the stipulated feedback function and initial state. Even worse, for $t > 2$, there are $t$-ary periodic sequences which cannot be generated by Algorithm 2 in \cite[Definition 2 on p. 131]{Golomb}. For example, if we view a binary de Bruijn sequence of order $n$ as a $t$-ary periodic sequence, then it cannot be generated. By the time the said algorithm terminates, some $(n-1)$-stage states must have been visited $t$ times, with $t>2$.
	
	\section{A Characterization for the GPO Algorithm to Produce de Bruijn Sequences}\label{sec:suffice}
	
	Among all of the binary periodic sequences that the GPO Algorithm can generate, our focus is on the de Bruijn ones. We now investigate the conditions for the algorithm to generate \textbf{only} de Bruijn sequences. For any order $n$, this is equivalent to ensuring that all $n$-stage states are visited as the algorithm runs its course. It is clear that $\cG_f$ contains at least one cycle or loop. To guarantee that the GPO Algorithm visits the states in a cycle or loop, the initial state must satisfy some conditions.
	
	\begin{lemma}\label{lemma-initial}
		For a given standard $f(x_0,x_1,\ldots,x_{n-1})$, let $\cC$ be a cycle or a loop in $\cG_f$. If the GPO Algorithm with input $(f,\bu)$ is to visit all states in $\cC$, then $\bu$ must also be in $\cC$.
	\end{lemma}
	
	\begin{proof}
		We prove the contrapositive of the statement. Suppose that $\bu$ is not in $\cC$. If $\cC$ is a loop, then the only state in $\cC$ is either $\0^n$ or $\1^n$. We start with the loop $\cC$ having $\0^n$ as the only state. Before reaching $\0^n$, the algorithm must have visited its predecessor $1\0^{n-1}$ first. We note that the other possible successor of $1\0^{n-1}$ is $\0^{n-1}1$, which has not appeared earlier because its two possible predecessors are $1\0^{n-1}$ and $\0^n$. Thus, the actual successor of $1\0^{n-1}$ must have been $\0^{n-1}1$, instead of $\0^n$, leaving the latter out. So $\0^n$ would never be visited by the GPO Algorithm and, hence, the generated sequence is not de Bruijn. The case of the loop $\cC$ having $\bu=\1^n$ as the only state can be similarly argued.
		
		If $\cC$ is a cycle, then let $\ba$ be the first state that the algorithm visits from among all of the states in $\cC$. Lemma \ref{lemma1} says that both of $\ba$'s children must have been visited before. One of the two, however, must also belong to $\cC$, contradicting the assumption that $\ba$ is the first. \qed
	\end{proof}
	
	Using the same method in the proof of Lemma \ref{lemma-initial}, we can directly infer the next result, whose proof is omitted here for brevity.
	
	\begin{lemma}\label{lemma-2}
		Let $f(x_0,x_1,\ldots,x_{n-1})$ be in the standard form. If there are $t > 1$ cycles or loops in $\cG_f$, then, taking an $n$-stage state in any one of the cycles or loops as the initial state, the GPO Algorithm will not visit any of the states belonging to the other cycles or loops.
	\end{lemma}
	
	The following lemma characterizes the input on which the GPO Algorithm is guaranteed to output a de Bruijn sequence.
	
	\begin{lemma}\label{lemma-1}
		On input a standard $f(x_0,x_1,\ldots,x_{n-1})$ and an $n$-stage state $\bu$, the GPO Algorithm generates a binary de Bruijn sequence of order $n$ if $\cG_f$ satisfies one of the following equivalent conditions.
		\begin{enumerate}
			\item There is a unique directed path from any state $\bv \neq \bu$ to $\bu$.
			\item There is a unique cycle or loop in $\cG_f$ containing $\bu$.
		\end{enumerate}
	\end{lemma}
	
	\begin{proof}
		Suppose that there is a unique directed path from any state $\bv \neq \bu$ to $\bu$ in $\cG_f$. Then any $\bv$ can be viewed as a descendent of $\bu$. Lemma \ref{lemma1} implies that by the time the algorithm revisits $\bu$, it must have visited $\bu$'s two children, if $\bu$ is in a cycle. If $\bu$ is in a loop instead, then it must have visited $\bu$'s only child. Applying the lemma recursively, the grandchildren of $\bu$ must have all been visited prior to that. Continuing the process, we confirm that all of $\bu$'s descendants must have been covered in the running of the algorithm. Thus, the generated sequence is de Bruijn of order $n$. \qed
	\end{proof}
	
	We combine Lemmas \ref{lemma-initial}, \ref{lemma-2} and \ref{lemma-1} to obtain the following result.
	
	\begin{theorem}\label{thm:m}
		Let a state $\bu=(u_0,u_1,\ldots,u_{n-1})$ and a standard $f(x_0,x_1,\ldots,x_{n-1})$ be given as the input. Then the GPO Algorithm generates a binary de Bruijn sequence of order $n$ if and only if there is a unique directed path from any state $\bv \neq \bu$ to $\bu$ in $\cG_f$. The condition is equivalent to the existence of a unique cycle or loop in $\cG_f$ with the property that $\bu$ is contained in this unique cycle or loop in $\cG_f$.
	\end{theorem}
	
	\begin{example}
		Let $f(x_0,\ldots,x_{n-1})=0$. The state graph $\cG_f$ contains only one loop $(0)$, from $\0^n$ to itself. All other states are descendants of $\0^n$. Taking $\bu=\0^n$, the GPO Algorithm produces the {\tt Prefer-One} de Bruijn sequence~\cite{Fred82}. The sequence is $(0000~1111~0110~0101)$ when $n=4$. 
		
		Similarly, let $f(x_0,\ldots,x_{n-1})=1$ and $\bu=\1^n$. Then the GPO Algorithm produces the {\tt Prefer-Zero} de Bruijn sequence \cite{Martin1934}, which is the complement of the {\tt Prefer-One} sequence. 
		\QEDB
	\end{example}
	
	It is now clear that, to generate a de Bruijn sequence by the GPO Algorithm, it suffices to find a pair $(f,\bu)$ that satisfies the requirement of Theorem \ref{thm:m}. This general task is both technically important and practically interesting. To see why Theorem~\ref{thm:m} is meaningful, we observe that the generation of the {\tt Prefer-Same} sequence on initial state $\0^n$ in~\cite{Fred82} is quite involved. Theorem~\ref{thm:m} provides an easier route. Indeed, taking the feedback function $f(x_0,\ldots,x_{n-1}) = x_{n-1} + 1$, the state graph has a unique cycle $(01)$. On input $f$ and either choice of initial $n$-stage states $01010\ldots$ and $10101\ldots$, the GPO Algorithm easily produces two distinct and complementary {\tt Prefer-Same} de Bruijn sequences. Theorem~\ref{thm:m} leads to the next two corollaries.
	
	\begin{corollary}\cite[Theorem 2]{Chang21}\label{cor-1}
		Let $n > m \geq 2$ and $h(x_0,x_1,\ldots,x_{m-1})$ be a feedback function whose FSR generates a de Bruijn sequence $\bs_m$ of order $m$. Let
		\[
		f(x_0,x_1,\ldots,x_{n-1}) := h(x_{n-m},x_{n-m+1},\ldots,x_{n-1})
		\]
		and $\bu$ be any $n$-stage state of $\bs_m$. Then the GPO Algorithm, on input $(f,\bu)$, generates a de Bruijn sequence of order $n$.
	\end{corollary}
	
	\begin{proof}
		The unique cycle in $\cG_f$ is $\bs_m$ and the initial state $\bu$ is in this cycle. \qed
	\end{proof}
	
	\begin{example}
		An order $4$ de Bruijn sequence $(0000~1001~1010~1111)$ is produced by the FSR with 
		$g(x_0,x_1,x_2,x_3) := 1+ x_0 + x_2 + x_3 + x_1 \cdot x_2 + x_1 \cdot x_3 + x_2 \cdot x_3 + x_1 \cdot x_2 \cdot x_3$. 
		Letting $n=6$ implies 
		$f(x_0,x_1,x_2,x_3,x_4,x_5) = 1+ x_2 + x_4 + x_5 + x_3 \cdot x_4 + x_3 \cdot x_5 + x_4 \cdot x_5 + x_3 \cdot x_4 \cdot x_5$. 
		Adding $\bb=000010$ to the input, the GPO Algorithm yields the de Bruijn sequence
		\[
		(00001010~00111011~00101101~11001111~11010010~00000110~00100110~10101111). ~\QEDB
		\]
	\end{example}
	
	\begin{corollary}\label{cor-2}
		Suppose that a standard $f_1(x_0,x_1,\ldots,x_{m-1})$ can be used in the GPO Algorithm to generate de Bruijn sequences of order $m$. Let $n > m$ and
		\[
		f_2(x_0,x_1,\ldots,x_{n-1}) := f_1(x_{n-m},x_{n-m+1},\ldots,x_{n-1}).
		\]
		Then $f_2$ can be used in the GPO Algorithm to generate de Bruijn sequences of order $n$.
	\end{corollary}
	
	\begin{proof}
		Theorem~\ref{thm:m} implies the existence of a unique cycle or loop in $\cG_{f_1}$. This, in turn, implies the existence of a unique cycle or loop in $\cG_{f_2}$. Let $\bu$ be any $n$-stage state in the cycle or loop in $\cG_{f_2}$. The conclusion follows by applying Theorem \ref{thm:m} to the input pair $(f_2,\bu)$. \qed
	\end{proof}
	
	\begin{remark}
		If $f_1$ is not standard, then Corollary \ref{cor-2} fails to hold in general. For a counterexample, let $f_1(x_0,\ldots,x_{m-1})=\prod_{i=0}^{m-1}x_i$. There are two loops, namely $(0)$ and $(1)$ in $\cG_{f_1}$. The GPO Algorithm on the input pair $(f_1,\0^n)$ produces the {\tt Prefer-One} de Bruijn sequence of order $m$. Let $n>m$ and $$f_2(x_0,x_1,\ldots,x_{n-1}) := f_1(x_{n-m},x_{n-m+1},\ldots,x_{n-1}) = \prod_{i=n-m}^{n-1}x_i.$$ Since the GPO Algorithm, on input $(f_2,\0^n)$, cannot visit $\1^n$, the resulting sequence is not de Bruijn. %\QEDB
	\end{remark}
	
	We have now achieved our first objective of providing a thorough treatment on when the GPO Algorithm generates de Bruijn sequences. As a consequence, we see that all prior ad hoc greedy algorithms, including those from preference functions and their respective generalizations, are special cases of the GPO Algorithm once the respective feedback function and initial state pairs are suitably chosen. The three families of sequences treated in~\cite[Section 3]{Chang21} and those produced from the input pairs listed in~\cite[Table 2]{Chang21} are examples of de Bruijn sequences whose construction fall into our general framework of {\tt Prefer-Opposite}. Here we offer another class of examples consisting of feedback functions that satisfy a simple number theoretic constraint on the indices of their variables.
	
	We look into the feedback function of the form
	\begin{equation}\label{equ1}
		f(x_0,x_1,\ldots,x_{n-1})= \prod_{j=1}^{n-1} x_j + x_k \cdot x_{\ell} \mbox{ for } 0 < k < \ell <n
	\end{equation}
	and characterize pairs $(k,\ell)$ such that the connected graph $\cG_f$ contains a unique loop $(0)$.
	
	\begin{proposition}\label{prop:int}
		There is a unique cycle, which is the loop $(0)$, in the state graph $\cG_f$ of the feedback function given in Equation (\ref{equ1}) if and only if $\gcd(n-k,\ell-k) = 1$.
	\end{proposition}
	\begin{proof}
		We define the following three feedback functions.
		\begin{align*}
			f_1(x_0,x_1,\ldots,x_{n-2}) &:= x_0 \cdot x_1 \cdots  x_{n-2} + x_{k-1} \cdot x_{\ell-1},\\
			f_2(x_0,x_1,\ldots,x_{n-2}) &:= x_{k-1} \cdot x_{\ell-1} \mbox{, and}\\
			f_3(x_0,x_1,\ldots,x_{n-k-1}) &:= x_{0} \cdot x_{\ell-k}.
		\end{align*}
		Vertices in $\cG_{f_1}$ and $\cG_{f_2}$ are $(n-1)$-stage states {whereas} those in $\cG_{f_3}$ are $(n-k)$-stage states. It is clear that the following {statements are equivalent}. 
		\begin{enumerate}
			\item The only cycle in $\cG_{f}$ is the loop $(0)$.
			\item The state graph $\cG_{f_1}$ has a unique cycle $(0)$.
			\item The state graph $\cG_{f_2}$ has exactly two cycles $(0)$ and $(1)$.
			\item The state graph $\cG_{f_3}$ has exactly two cycles $(0)$ and $(1)$.
		\end{enumerate}
		Hence, it suffices to establish that $\cG_{f_3}$ has exactly the two cycles $(0)$ and $(1)$ if and only if $\gcd(n-k,\ell-k)=1$.
		
		It is immediate to verify that the two loops $(0)$ and $(1)$ are in $\cG_{f_3}$. We show that there exists another cycle in it if and only if $\gcd(n-k,\ell-k)>1$.
		
		The output sequence ${\bf s}=s_0,s_1,s_2,\ldots$ of the $(n-k)$-stage FSR with feedback function $f_3$ has $s_{i+n-k}=1$ if and only if $s_i=s_{i+\ell-k}=1$, for any $i \geq 0$. Suppose that there is another cycle $C=(a_0,a_1,\ldots,a_{N-1})$ in $\cG_{f_3}$ and $N>1$ is its least period.
		If $a_i=1$ for some integer $i$, then we have
		\[
		1=a_i=a_{i-(n-\ell)}=a_{i-(n-k)}=a_{i-2(n-\ell)}=a_{i-2(n-k)}=\cdots.
		\]
		If $a_i=1$, then we can deduce inductively that
		\begin{equation}\label{eq:aa}
			1=a_i=a_{i-t(n-\ell)}=a_{i-t(n-k)} \mbox{ for all } 
			t \in \{0,1,2,\ldots\}.
		\end{equation}
		It is immediate to confirm that if $a_i=0$, then $a_{i-t(n-\ell)}=a_{i-t(n-k)}=0$, since $a_{i-t(n-\ell)}=1$ or $a_{i-t(n-k)}=1$ implies $a_i=1$ by Equation (\ref{eq:aa}).
		So the cycle $C$ must satisfy
		\[
		a_i=a_{i-t(n-\ell)}=a_{i-t(n-k)}   \mbox{ for all } t = 0,1,2,\ldots,
		\]
		which means that the period $N$ divides $n-k$ and $n-\ell$ simultaneously. Thus
		\[
		\gcd(n-k,n-\ell)=\gcd(n-k,\ell-k) > 1.
		\]
		
		Conversely, if $\gcd(n-k,\ell-k) = t > 1$, then $\cG_{f_3}$ contains the cycle
		\[
		(1\0^{t-1}\, 1\0^{t-1}\, \ldots \, 1\0^{t-1}),
		\]
		establishing the existence of at least another cycle when $\gcd(n-k,\ell-k) > 1$. \qed
	\end{proof}

	Whenever $\gcd(n-k,\ell-k)=1$, the GPO Algorithm can use the feedback function $f$ in Equation (\ref{equ1}) as an input. Otherwise, the function can be used as an input in the GJPO Algorithm, to be introduced later. A combinatorial identity connects $\gcd(n-k,\ell-k)$ with the number of cycles in $\cG_f$. The number of cycles is the number of cyclic decomposition of $\gcd(n-k,\ell-k)$. For example, when $n=11$ with $(k,\ell) = (1,6)$, we have $\gcd(10,5)=5$, so there are $7$ cycles. The details can be found, \eg, in Sequence A008965 in OEIS~\cite{OEIS65}.
	
	The next result generalizes Proposition~\ref{prop:int}. The proof follows similarly and, thus, is omitted.
	
	\begin{theorem}
		The state graph of the FSR with feedback function
		\begin{equation}\label{eq:genprod}
			f(x_0,x_1,\ldots,x_{n-1}) = \prod_{j=1}^{n-1} x_j + x_{k_1} \cdot x_{k_2} \cdots x_{k_t} \mbox{ for } 0 < k_1 < k_2 < \ldots < k_t < n,
		\end{equation}
		has a unique loop $(0)$ if and only if $\gcd(n-k_1,k_2-k_1,\ldots,k_t-k_1)=1$.
	\end{theorem}
	
	\section{Graph Joining and the GPO Algorithm}\label{sec:GJ}
	
	A typical standard feedback function $f(x_0,\ldots,x_{n-1})$ often fails to satisfy the condition in Theorem \ref{thm:m}, \ie, its $\cG_f$ tends to decompose into $t \geq 2$ {\it components}, each with a unique cycle or loop. Recall that the components of any non-null graph in its decomposition are non-null detached subgraphs, no two of them share any edge or vertex in common. The graph is then said to be a {\it union} of its components. When $\cG_f$ has $t \geq 2$ components, the sequence produced by the GPO Algorithm is not de Bruijn. This section introduces a new method of joining the components to turn the outputs into de Bruijn sequences.
	
	Suppose that, for a given standard $f$, its $\cG_f$ has components $G_1,G_2,\ldots,G_t$. Let $C_i$ be the unique cycle or loop in $G_i$ for each $1 \leq i \leq t$. Without loss of generality, let us take any $n$-stage state $\bu$ in $C_1$ as the initial state to run the GPO Algorithm. Lemmas~\ref{lemma-initial} and ~\ref{lemma-2} tell us that, as the sequence is being generated, the algorithm visits all of the states in $G_1$ but none of the states in the cycles or loops $C_2,\ldots,C_t$. We now modify the algorithm so that it can continue to cover all of the remaining states in $\cG_f$ before revisiting $\bu$.
	
	\begin{lemma}\label{lemma-m}
		For a given standard $f$, let $G_1,G_2,\ldots,G_t$ be the components in $\cG_f$. Let $C_i$ be the unique cycle or loop in $G_i$, for each $i$. Let $(f,\bu)$ be the input of the GPO Algorithm where $\bu$ is a state in $C_1$. Let $\bw$ in $C_j$ be an $n$-stage state such that its companion state $\widetilde{\bw}$ is a leaf in $G_1$, for some $j \in \{2,3,\ldots,t\}$. Let $\bv$ be the child of $\bw$ that does not belong to $C_j$. Then the GPO Algorithm can be modified to visit all of the states in $G_1$ and in $G_j$ if we assign $\bw$ to be the successor of $\bv$.
	\end{lemma}
	
	\begin{proof}
		Prior to the additional assignment rule, starting from $\bu$ in $C_1$, the GPO Algorithm visits all states in $G_1$ but cannot cover any of the states in $C_2 \cup C_3 \cup \ldots \cup C_t$.
		
		Suppose that $C_j$ is a cycle. Note that $\bw$ has two children, namely $\bv$, which does not belong to $C_j$, and $\widehat{\bv}$, which is in $C_j$. As the algorithm runs, let $\widetilde{\bw}$ be its currently visited state. Since $\widehat{\bv} \in C_j$, the algorithm cannot reach it by Lemma \ref{lemma-2}. Hence, the actual predecessor of $\widetilde{\bw}$ in the algorithm's output must be $\bv$. The new assignment rule, however, forces $\bw$ to be the successor of $\bv$.
		It is easy to check that after adding this assignment, the algorithm also terminates when it revisits $\bu$ and any other $n$-stage state is visited at most once. To guarantee that the modified algorithm revisits $\bu$, all the states in $G_1$ must be visited, including $\widetilde{\bw}$. Its unique possible predecessor has to be $\widehat{\bv}$, because the new assignment rule forces $\bv$ to be the actual predecessor of $\bw$. So we have deduced that $\widehat{\bv}\in C_j$ must be visited by the algorithm. Applying Lemma \ref{lemma1}, repeatedly if necessary, all states in $G_j$ can be visited by the algorithm. Thus, all states in $G_1$ and $G_j$ can be visited by the modified algorithm.
		
		If $C_j$ is a loop, then $\bw$ and $\bv$ are conjugate states. Adding the rule that assigns $\bw$ to be the actual successor of $\bv$ makes $\widetilde{\bw}$ the successor of $\bw$ in the output sequence. Once $\bw$ is visited, $\bv$ must have been visited, and all states in $G_j$ can be reached by the modified algorithm. \qed
	\end{proof}
	
	The selection of states, including $\bw$, $\widetilde{\bw}$, $\bv$ and $\widehat{\bv}$, in the above proof will be illustrated in Figure~\ref{fig:one} of Example \ref{ex:gjoin}. The first state in $C_j$ that the modified algorithm visits must be $\bw$. By Lemma \ref{lemma1}, the {\it order of appearance} of the states in $C_j$ in the output sequence follows the {\it direction of their edges}. The next state visited after $\bw$ is the parent of $\bw$ in $C_j$, and so on. The last to be visited is the child $\widehat{\bv}$ of $\bw$. This fact implies that distinct valid options for $\bw$ in $C_j$ lead to shift-inequivalent output sequences when the initial state $\bu$ is fixed. Keeping everything fixed but changing the initial state to any other state in the same cycle $C_1$, however, may produce shift-equivalent output sequences. Such instances of collisions, albeit being rare, should be taken into account in the enumeration of inequivalent output sequences.
	
	The process detailed in the proof of Lemma \ref{lemma-m} explains the thinking behind our modification of the GPO Algorithm. Suppose that the modified algorithm can now cover all of the states in two subgraphs, say $G_1$ and $G_2$. We identify a state $\bx$ in $C_3$ such that $\widetilde{\bx}$ is a leaf in $G_1$ or $G_2$. If $\by$ is the child of $\bx$ that does not belong to $C_3$, then we add the rule that assign $\bx$ to be the successor of $\by$ to ensure that all states in $G_3$ can be visited. If such an assignment can be done to ``join'' $G_1,G_2,\ldots,G_t$ together, then the resulting sequence of the modified algorithm is de Bruijn of order $n$.
	
	We outline the process of graph joining, from a given state graph of a standard feedback function, in the following steps.
	\begin{labeling}{\bf Generate}
		\item[{\bf Step} $1$] Choose an arbitrary component, say $G_1$. Set $\Omega \gets \{G_1\}$.
		\item[{\bf Step} $2$] Find a component, say $G_2$, such that $\bw_2$ is a state in $C_2$, $\widetilde{\bw_2}$ is a leaf of $G_1$, and $\bv_2$ is a child of $\bw_2$ which is not in $C_2$. Set $\Omega \gets \Omega\cup\{G_2\}$.
		\item [{\bf Step} $3$] Continue inductively by finding a $G_i$ such that $\bw_i$ is a state in $C_i$, $\widetilde{\bw_i}$ is a leaf of some subgraph in $\Omega$, and $\bv_i$ is a child of $\bw_i$ which is not in $C_i$. Set $\Omega \gets \Omega\cup\{G_i\}$.
		\item [{\bf Step} $4$] End if $\Omega$ contains all of the component graphs. Otherwise declare failure since the output will not be de Bruijn.
		\item [{\bf Generate}] Run the GPO Algorithm on an arbitrary state $\bu$ in $C_1$ as the initial state. Upon reaching $\bv_i$ for any $i=2,\ldots,t$ as the current state, the usual procedure in the algorithm is modified such that the next state is assigned to be $\bw_i$. In all other occasions, comply with the rules of the algorithm to generate the next state. Output a de Bruijn sequence.
	\end{labeling}
	
	\begin{example}\label{ex:gjoin}
		Let $f(x_0,x_1,x_2,x_3)=x_1 + x_2 \cdot x_3$. Then $\cG_f$ is divided into three components, namely, $G_1$ with $C_1=(1110)$, $G_2$ with $C_2=(010)$, and $G_3$ with $C_3=(0)$. Figure~\ref{fig:one} shows $\cG_f$ with a configuration that allows for the merging of the three components using the added assignments.
		
		\begin{figure*}[h!]
			\centering
			\begin{tikzpicture}
				[
				> = stealth,
				shorten > = 3pt,
				auto,
				node distance = 2.45cm,
				semithick
				]
				
				\tikzstyle{every state}=
				\node[rectangle,fill=white,draw,rounded corners,minimum size = 4mm]
				%%% G1	
				\node[state,fill=lightgray] (1) {$1101\,{\circled{8}}$};
				\node[state,fill=lightgray] (2) [right of=1] {$1011\,{\circled{14}}$};
				\node[state,fill=lightgray] (3) [below of=2] {$0111\,{\circled{15}}$};
				\node[state,fill=lightgray] (4) [below of=1] {$1110=\bu\,{\circled{1}}$};
				
				\node[state] (5) [above of=1] {$0110\,{\circled{7}}$};
				\node[state] (6) [below of=4] {$1111\,{\circled{16}}$};
				\node[state] (7) [right of=3] {$0011\,{\circled{6}}$};
				\node[state,fill=blue!20] (8) [above of=2] {$0101=\widetilde{\bw}\,{\circled{13}}$};
				%%%%% G2
				\node[state,fill=blue!20] (9) [right of=2] {$0100=\bw\,{\circled{10}}$};
				\node[state,fill=lightgray] (10) [right of=9] {$1001\,{\circled{11}}$};
				\node[state,fill=lightgray] (11) [below of=10] {$0010=\widehat{\bv}\,{\circled{12}}$};
				\node[state] (12) [above of=9] {$1010=\bv\,{\circled{9}}$};
				\node[state] (13) [above of=10] {$1100\,{\circled{2}}$};
				\node[state,fill=blue!20] (14) [right of=11] {$0001=\widetilde{\bx}\,{\circled{5}}$};
				%%%%% G3
				\node[state] (15) [right of=13] {$1000=\by\,{\circled{3}}$};
				\node[state,fill=blue!20] (16) [below of=15] {$0000=\bx\,{\circled{4}}$};
				
				\path[->] (1) edge (2);
				\path[->] (2) edge (3);
				\path[->] (3) edge (4);
				\path[->] (4) edge (1);
				
				\path[->] (5) edge (1);
				\path[->] (8) edge (2);
				\path[->] (7) edge (3);
				\path[->] (6) edge (4);
				
				\path[->] (9) edge (10);
				\path[->] (10) edge (11);
				\path[->] (11) edge (9);
				
				\path[->] (12) edge (9);
				\path[->] (13) edge (10);
				\path[->] (14) edge (11);
				
				\path[->] (15) edge (16);
				\path[->] (16) edge[loop below] (16);
			\end{tikzpicture}
			\caption{The state graph $\cG_{f}$ for $f(x_0,x_1,\ldots,x_3) = x_1+ x_2 \cdot x_3$. The components $G_1,G_2,G_3$ are from left to right. The encircled label indicates the state's order of appearance in the generated de Bruijn sequence when $\bu$ is the initial state.}
			\label{fig:one}
		\end{figure*}
		
		Taking as $\bu$ the respective four states in $C_1$, namely $1110,1101,~1011,~0111$, we obtain three inequivalent de Bruijn sequences. On $\bu=1110$, the states are visited in the following order. We use $\Rightarrow$ to indicate that the state transition is governed by an additional assignment rule.
		\begin{align*}
			\bu=& 1110 \,\rightarrow \, 1100 \, \rightarrow \, 1000 \, \Rightarrow \, 0000 \, \rightarrow \, 0001 \, \rightarrow \, 0011 \, \rightarrow \, 0110 \,\rightarrow \, 1101 \, \rightarrow\\
			&1010 \, \Rightarrow \, 0100 \, \rightarrow \, 1001 \, \rightarrow \, 0010
			\, \rightarrow \, 0101 \, \rightarrow \, 1011 \, \rightarrow \, 0111 \, \rightarrow \, 1111 \, \rightarrow \bu.
		\end{align*}
		The output sequence is $(1110~0001~1010~0101)$. Choosing $\bu=0111$, the output sequence $(0111~1000~0110~1001)$ is shift-equivalent to the previous one, \ie, here there is a collision. The other two inequivalent output sequences are, respectively, $(1101~0000~1100~1011)$ on $\bu=1101$ and $(1011~0000~1111~0100)$ on $\bu=1011$.
		
		Suppose that we change $\bw$ from $0100$ to $0010$, making $\widetilde{\bw}=0011$. The output sequences $(1110~0001~0100~1101)$ and $(0111~1000~0101~0011)$, on respective initial states $1110$ and $0111$, are again shift-equivalent. The other two outputs are $(1011~0000~1010~0111)$ on $\bu=1011$ and $(1101~0110~0001~0011)$ on $\bu=1101$. In total, we have generated six shift-inequivalent de Bruijn sequences of order $4$ in this example. \QEDB
	\end{example}

\begin{example}\label{ex:prefer-opp}
The original generating algorithm for the {\tt Prefer-Opposite} de Bruijn sequence in \cite{Alh10} is rather complicated. We provide a very simple alternative by using the above results, with $f(x_0,x_1,\ldots,x_{n-1}) := x_{n-1}$. There are two components in $\cG_f$, namely, $G_1$ with $C_1=(0)$ and $G_2$ with $C_2=(1)$. Directly using the GPO Algorithm with $f$ on initial state $\mathbf{0}^n$ ($\mathbf{1}^n$, respectively), the output sequence contains all $n$-stage states except for $\mathbf{1}^n$ ($\mathbf{0}^n$, respectively). When $0\mathbf{1}^{n-1}$ ($1\mathbf{0}^{n-1}$, respectively) is visited, we assign $\mathbf{1}^n$ ($\mathbf{0}^n$, respectively) to be its successor. Then the resulting sequence is de Bruijn. When we take $\mathbf{0}^n$ as the initial state, the corresponding algorithm is Algorithm \ref{algo:po}. \QEDB
\end{example}
			
\begin{algorithm}[ht!]
\caption{{\tt Prefer-Opposite}}
\label{algo:po}
\begin{algorithmic}[1]
\renewcommand{\algorithmicrequire}{\textbf{Input:}}
\renewcommand{\algorithmicensure}{\textbf{Output:}}
\Require Initial state ${\bf 0}^n$.
					\Ensure A de Bruijn sequence.
					\State{$\bc:=c_0, c_1, \ldots ,c_{n-1} \gets {\bf 0}^n$}	
					\Do
					\State{Print($c_0$)}
					\If{$\bc=0{\bf 1}^{n-1}$}
					\State{$\bc \gets {\bf 1}^n$}
					\Else
					\If{$c_1, c_2, \ldots, c_{n-1}, \overline{c_{n-1}}$ has not appeared before}
					\State{$\bc \gets c_1, c_2, \ldots, c_{n-1},  \overline{c_{n-1}}$} %\label{Algpo:line5}
					\Else
					\State{$\bc \gets c_1, c_2, \ldots, c_{n-1}, c_{n-1} $} %\label{Algpo:line7}
					\EndIf
					\EndIf
					\doWhile{$\bc \neq {\bf 0}^n$}
\end{algorithmic}
\end{algorithm}

	Since the underlying idea of our graph joining method (GJM) is similar to that of the cycle joining method (CJM), we briefly recall how the CJM works. The state graph of an FSR with a {\bf nonsingular} feedback function $h(x_0,x_1,\ldots,x_{n-1})$ is divided into disjoint cycles. If an $n$-stage state $\bw$ and its companion state $\widetilde{\bw}$ belong to two distinct cycles, then the cycles can be joined by exchanging the predecessors of $\bw$ and $\widetilde{\bw}$. Such pair of cycles are said to be {\it adjacent} and the pair $(\bw,\widetilde{\bw})$ is called a {\it companion pair}. The process is repeated to join the rest of the cycles by identifying suitable companion pair(s). If all of the initially disjoint cycles can be joined together, then the final cycle must be a de Bruijn sequence of order $n$.
	
	For an FSR with nonsingular feedback function $h$, its {\it adjacency graph} is an \emph{undirected} multigraph whose vertices correspond to the cycles in the corresponding state graph. There exists an edge between two vertices if and only if they share a companion pair. The number of shared companion pairs labels the edge. The number of resulting inequivalent de Bruijn sequences is the number of spanning trees in the adjacency graph.
	
	We use similar concepts to formally define our GJM. Given a standard $f$, suppose that $\cG_f$ consists of $t$ components $G_1,\ldots,G_t$. Let $(\bw,\widetilde{\bw})$ be a companion pair with $\bw$ a state in the cycle (or loop) $C_i$ of $G_i$ and $\widetilde{\bw}$ a leaf in $G_j$ for $1 \leq i \neq j \leq t$. Companion pairs play an important role in our new method because they can modify the GPO Algorithm to visit all states in $G_i$ and $G_j$ when the initial state is one of the state(s) in $C_j$. We call such $(\bw,\widetilde{\bw})_{i,j}$ pair a {\it preference companion pair} (PCP) \textbf{from} $G_i$ \textbf{to} $G_j$. Each state in the pair is the other's {\it preference companion state}. Two components $G_i$ and $G_j$ are {\it adjacent} if they share a PCP $(\bw,\widetilde{\bw})_{i,j}$. The pair joins adjacent components into one.
	
	\begin{definition}\label{def:adjgraph}
		The preference adjacency graph (PAG) for an FSR with a standard feedback function $f$ is a \emph{directed multigraph} $\mathbb{G}_f$ whose $t$ vertices correspond to the $t$ components $G_1,\ldots,G_t$ in $\cG_f$. There exists a directed edge from $G_i$ to $G_j$ if and only if there exists a PCP $(\bw,\widetilde{\bw})_{i,j}$. The number of directed edges, each identified by a PCP, from $G_i$ to $G_j$ is the number of their PCPs.
	\end{definition}
	
	By definition, $\mathbb{G}_f$ contains no loops. A {\it rooted spanning tree} in $\mathbb{G}_f$ is a rooted tree that contains all of $\mathbb{G}_f$'s vertices. We use $\Upsilon_{\mathbb{G}}$, or simply $\Upsilon$ if $\mathbb{G}_f$ is clear, to denote a rooted spanning tree in $\mathbb{G}_f$.
	By Lemma~\ref{lemma-m}, once we obtain a rooted spanning tree in $\mathbb{G}_f$, with the root being in some cycle $C_k$, then we can run the GPO Algorithm with an arbitrary state in $C_k$ as the initial state. Suppose that the following three conditions are satisfied.
	\begin{enumerate}
		\item The current state is $\bv$, which is the predecessor of $\bw$.
		\item The state $\bv$ is not in the cycle that contains $\bw$.
		\item The pair $(\bw,\widetilde{\bw})_{i,j}$ is an edge in the spanning tree.
	\end{enumerate}
	Then we assign $\bw$ to be the successor of $\bv$ to ensure that the resulting sequence is de Bruijn.
	
	We modify the GPO Algorithm by adding the GJM in the assignment rule, guided by the relevant PCPs. We call the modified algorithm the {\tt Graph Joining Prefer -Opposite} (GJPO)  Algorithm, presented here as Algorithm~\ref{algo:new}. The algorithm sets out to find all preference companion pairs, before determining all possible rooted spanning trees. Once this is done and if there exists at least one rooted spanning tree, then we can choose one such tree, either deterministically or randomly, and generate the corresponding Bruijn sequence.
	
	\begin{algorithm}[t!]
		\caption{{\tt Graph Joining Prefer-Opposite} (GJPO)}
		\label{algo:new}
		\begin{algorithmic}[1]
			\renewcommand{\algorithmicrequire}{\textbf{Input:}}
			\renewcommand{\algorithmicensure}{\textbf{Output:}}
			\Require A standard feedback function $f(x_0,x_1,\ldots,x_{n-1})$.
			\Ensure de Bruijn sequences.
			\State{Construct the state graph $\cG_f$}\label{line1}
			\State{$V_{\mathbb{G}} \gets \{G_1,G_2,\ldots,G_t\}$ and $\cC \gets \{C_1,C_2,\ldots,C_t \}$} \Comment{$V_{\mathbb{G}}$ is vertex set of $\mathbb{G}$}\label{line2}
			\For{$i$ from $1$ to $t$} \Comment{Populate the edges in $\mathbb{G}$}
			\For{every state $\bw \in C_i$}
			\If{$\widetilde{\bw}$ is a leaf in $C_j$ for $j \neq i$}
			\State{add edge in $\mathbb{G}$ from $G_i$ to $G_j$ labelled by $\bw$}
			\EndIf
			\EndFor
			\State{$K_{i,j} \gets$ the list of PCPs from $G_i$ to $G_j$ for $i \neq j$}\label{line7}
			\EndFor
			\State{Derive the simplified undirected graph $\mathbb{H}$ from $\mathbb{G}$: The vertex set of $\mathbb{H}$ is the vertex set of $\mathbb{G}$. If $K_{i,j}$ is nonempty for $1 \leq i \neq j \leq t$, then add an undirected edge $e_{i,j}$ in $\mathbb{H}$.}\label{line8}
			\State{Generate the set $\Gamma_{\mathbb{H}}$ of all spanning trees in $\mathbb{H}$} \Comment{We use~\cite[Algorithm~5]{Chang2019}} \label{line9}
			\For{each $\Upsilon \in \Gamma_{\mathbb{H}}$} \label{line10}%\Comment{Generate all rooted directed spanning trees in $\mathbb{G}$}
			\State{$E_{\Upsilon} \gets$ the set $\{e_1,e_2,\ldots,e_{t-1}\}$ of edges in $\Upsilon$}
			\For{$k$ from $1$ to $t-1$}
			\For{each of the two possible directions for $e_k$} \Comment{from $G_i$ to $G_j$ or $G_j$ to $G_i$}
			\If{the corresponding $K_{i,j}$ (respectively $K_{j,i}$) is nonempty}
			\State{choose each of its element, in sequence, as an edge in the directed tree}
			\EndIf
			\EndFor
			\EndFor
			\State{$\Sigma_{\Upsilon} \gets$ all directed trees in $\mathbb{G}$ that corresponds to $\Upsilon$}
			\For{each $\Lambda \in \Sigma_{\Upsilon}$}
			\If{$\Lambda$ is not a rooted tree}
			\State{$\Sigma_{\Upsilon} \gets \Sigma_{\Upsilon} \setminus \Lambda$}
			\EndIf
			\EndFor
			\EndFor
			\State{$\Omega \gets \bigcup_{\Upsilon \in \Gamma_{\mathbb{H}}} \Sigma_{\Upsilon}$} \Comment{The set of all rooted spanning trees in $\mathbb{G}$}\label{line20}
			\For{each $\Delta \in \Omega$}\label{line21}
			\State{$E_{\Delta} \gets$ set of states that label the directed edges in $\Delta$}
			\For{each state $\bu$ in $\cC$ in the root component $G$ of $\Delta$}
			\State{$\bc:=c_0, c_1, \ldots ,c_{n-1} \gets \bu$}	
			\Do
			\State{Print($c_0$)}
			%	    	\State{$y \gets f(c_0,c_1,\ldots,c_{n-1})$}
			\If{$c_1, c_2, \ldots, c_{n-1}, f(c_0,c_1,\ldots,c_{n-1}) \notin E_{\Delta}$}
			\If{$c_1, c_2, \ldots, c_{n-1}, \overline{f(c_0,c_1,\ldots,c_{n-1})}$ has not appeared before}
			\State{$\bc \gets c_1, c_2, \ldots, c_{n-1},  \overline{f(c_0,c_1,\ldots,c_{n-1})}$}
			\Else
			\State{$\bc \gets c_1, c_2, \ldots, c_{n-1}, f(c_0,c_1,\ldots,c_{n-1}) $}
			\EndIf
			\Else
			\State{$\bc \gets c_1, c_2, \ldots, c_{n-1}, f(c_0,c_1,\ldots,c_{n-1})$} \Comment{The added assignment}
			\State{$E_{\Delta} \gets E_{\Delta} \setminus \{\bc\}$}
			\EndIf
			\doWhile{$\bc \neq \bu$}
			\EndFor
			\EndFor
		\end{algorithmic}
	\end{algorithm}
	
	To analyze its complexity, we break the GJPO Algorithm down into five stages.
	\begin{labeling}{{\bf Stage} $5$}
		\item[{\bf Stage} $1$] Described in Line~\ref{line1}, the state graph $\cG_{f}$ is built in $O(n)$ time. In its most basic implementation, storing this graph is not efficient, requiring $O(2^n)$ space. One can reduce this memory requirement by storing, per graph components, only its cycle and a list of leaves. The saving varies, depending on $f$, and is not so easy to quantify. Checking whether a vertex state $c_1,c_2,\ldots,c_{n-1},y$, for some $y \in \F_2$, is a non-leaf is very efficient since one only needs to determine if evaluating the algebraic normal form on inputs $x,c_1,c_2,\ldots,c_{n-1}$, for both $x \in \F_2$, yields the specified state.
		
		\item[{\bf Stage} $2$] The procedure in Lines~\ref{line2} to~\ref{line7} builds the PAG $\mathbb{G}$. For each component graph, the number of operations to perform is roughly $(t-1)$ times the period of its cycle, giving the total time estimate to be $(t-1)$ times the sum of the periods of the $t$ cycles. Storing $K_{i,j}$ for $1 \leq i \neq j \leq t$ needs $O(t^3)$ space since there are $t^2-t$ entries, each having at most $\max\{|K_{i,j}|\} \leq t-1$ elements.
		
		\item [{\bf Stage} $3$] Once $\mathbb{G}$ has been determined, the procedure in Line~\ref{line8} derives a simpler graph $\mathbb{H}$, taking $O(t^2)$ in both time and storage, to use as an input in an intermediate step to eventually identify all rooted spanning trees in $\mathbb{G}$. In the simpler graph, the edge direction is removed and multiple edges are collapsed into a single edge. This is done since generating all spanning trees in a graph is quite costly. The time requirement, using~\cite[Algorithm~5]{Chang2019}, is in the order of the number of spanning trees in the input graph. The simpler $\mathbb{H}$ is faster to work on. 
		Line~\ref{line9} continues the process by generating all spanning trees in $\mathbb{H}$.
		
		\item [{\bf Stage} $4$] The steps given in Lines~\ref{line10} to~\ref{line20} combine $\Gamma_{\mathbb{H}}$ and the lists of PCPs stored in $K_{i,j}$, for all applicable $i$ and $j$, to generate a set $\Omega$ of all rooted spanning trees in the PAG graph $\mathbb{G}$. The time complexity to generate $\Omega$ is $O(|\Gamma_{\mathbb{H}}|) \cdot 2^{t-1} \cdot t^3$. Storing it takes $O(|\Gamma_{\mathbb{H}}|) \cdot t$ space.
		
		\item [{\bf Stage} $5$] The routine described in Line~\ref{line21} onward generates an actual de Bruijn sequence per identified rooted spanning tree. This process is negligible in resource requirements.
	\end{labeling}
	Our presentation and analysis above are geared toward completeness, not practicality. \textit{Generating the entire set $\Omega$ is often neither necessary nor desirable}. In actual deployment, computational routines to identify only a required number of rooted spanning trees can be carefully designed for speed and space savings.
	
	\begin{theorem}\label{thm3}
		Given a standard feedback function $f(x_0,\ldots,x_{n-1})$, Algorithm~\ref{algo:new} outputs de Bruijn sequences of order $n$ if and only if the set $\Gamma_{\mathbb{H}}$ in Line~\ref{line9} is nonempty.
	\end{theorem}
	
	\begin{proof}
		The correctness of Algorithm~\ref{algo:new} follows from Lemma \ref{lemma-m}. \qed
	\end{proof}
	
	Suppose that, for a chosen PCP $(\bw,\widetilde{\bw})_{i,j}$ that connects $G_i$ to $G_j$, $\bw$ is the first state in $C_i$ that Algorithm~\ref{algo:new} visits. Then the following result is a direct consequence of Theorem~\ref{thm3}.
	
	\begin{corollary}
		The number of distinct rooted trees is a lower bound for the number of inequivalent de Bruijn sequences produced.
	\end{corollary}
	
	Writing a basic implementation in {\tt python 2.7} and feeding it many standard feedback functions lead us to an interesting computational observation. For the majority of the input functions, distinct choices of initial states lead to inequivalent de Bruijn sequences, for a fixed rooted spanning tree. We harvest many more valid sequences than the number of distinct rooted spanning trees in $\mathbb{G}_f$. This suggests the following problem.
	
	\begin{problem}
		Let a standard feedback function $f$, or a class $\mathcal{F}$ of standard feedback functions, be given. Provide a closed formula, or a good estimate, of the number of inequivalent outputs.
	\end{problem}
	
	We require the spanning trees to use in Algorithm~\ref{algo:new} to be the rooted ones only. Otherwise, there will be a cycle $C_j$, for some $1 \leq j \leq t$, whose states are never visited. This clearly implies that the resulting sequence fails to be de Bruijn.
	
	\begin{example}\label{ex:nonrooted}
		We use the setup in Example~\ref{ex:gjoin} and refer to Figure~\ref{fig:one} for the state graph $\cG_{f}$. Note that the tree in Figure~\ref{fig:nonrooted} is a non-rooted spanning tree in $\mathbb{G}_f$. We use the PCPs $(0100,0101)_{2,1}$ and $(1001,1000)_{2,3}$ as the directed edges in the tree.
		
		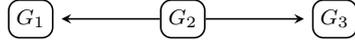
\begin{figure}[h!]
			\centering
			\begin{tikzpicture}
				[
				> = stealth,
				shorten > = 3pt,
				auto,
				node distance = 2cm,
				semithick
				]
				
				\tikzstyle{every state}=
				\node[rectangle,fill=white,draw,rounded corners,minimum size = 5mm]
				
				\node[state] (1) {$G_1$};
				\node[state] (2) [right of=1] {$G_2$};
				\node[state] (3) [right of=2] {$G_3$};
				\path[->] (2) edge (1);
				\path[->] (2) edge (3);	
			\end{tikzpicture}
			\caption{A Non-rooted Spanning Tree in $\mathbb{G}_f$ for $f(x_0,x_1,\ldots,x_3) = x_1+ x_2 \cdot x_3$.}\label{fig:nonrooted}
		\end{figure}
		\noindent We choose $\bu=1011 \in C_1$ as the initial state and run the modified version of the GPO Algorithm, with the added assignment rules that the successor of $1100$ must be $1001$, and the successor of $1010$ must be $0100$. The algorithm will never visit $C_3=(0)$, producing the output $(1011~0011~1101~000)$. Now, if we replicate the process with $\bu=0000 \in C_3$, then the output sequence is $(0000~1100~1010)$, leaving the four states in $C_1=(1011)$ out. \QEDB
	\end{example}
	
	\section{The Preference Adjacency Graph of a Class of Feedback Functions}\label{sec:PAG}
	
	For an arbitrary standard feedback function $f$, determining $\cG_f$, \eg, identifying its cycle(s) or loop(s), is generally not an easy task. Some feedback functions cannot be used by Algorithm~\ref{algo:new} to generate de Bruijn sequence, as illustrated by the next example.
	\begin{example}\label{ex:nogo}
		Let $n \geq 4$ and 
		\[
		f(x_0,x_1,\ldots,x_{n-1}) :=x_{n-3} \cdot x_{n-2} + x_{n-3} \cdot x_{n-1} + x_{n-2} \cdot x_{n-1}.
		\]
		The state graph $\cG_f$ has two loops, one in each of the two components, namely $C_1=(0)$ in $G_1$ and $C_2=(1)$ in $G_2$. The companion state $\0^{n-1}1$ of $\0^{n} \in C_1$ is in $G_1$ and, similarly, the companion state $\1^{n-1}0$ of $\1^{n} \in C_2$ is in $G_2$. Hence, we cannot join the two components. \QEDB
	\end{example}
	
	In this section we examine the PAG of a class of feedback functions. Our main motivation is to identify connections between structures of different orders that help in building the PAG. Let $1 \leq m<n$ be integers. Let $h(x_0,\ldots,x_{m-1})$ be any given nonsingular feedback function. The preference adjacency graph $\mathbb{G}_{F_h}$ of $F_h$, which is given by
	\begin{equation}\label{equ:fh}
		F_h(x_0,x_1,\ldots,x_{n-1}): = x_{n-m}+g(x_{n-m+1},\ldots,x_{n-1})=h(x_{n-m},\ldots,x_{n-1}),
	\end{equation}
	becomes easier to determine. Algorithm~\ref{algo:new} can then use $F_h$ to generate de Bruijn sequences.
	
	Since $h$ is nonsingular, the $t$ components, say $C_{h,1},C_{h,2},\ldots,C_{h,t}$, of $\cG_h$ are cycles. The states in $\cG_h$ are $m$-stage. These $t$ cycles can be joined into a de Bruijn sequence of order $m$ by the cycle joining method if there exists enough companion pairs $(\bw,\widetilde{\bw})$ with $\bw \in C_{h,k}$ and $\widetilde{\bw} \in C_{h,\ell}$ between appropriate pairs $1 \leq k \neq \ell \leq t$.
	
	Let us now consider $\cG_{F_h}$, whose states are $n$-stage. Its $t$ components are $G_1, G_2, \ldots$, $G_t$. For $1 \leq i \leq t$, let $C_i$ be the unique cycle or loop in $G_i$, labelled based on the one-to-one correspondence between $C_{h,i}$ and $C_i$ induced by how $F_h$ is defined. The component subgraph $G_i$ can be divided further into distinct rooted trees when the edges that connect the states in $C_i$ are deleted. The respective roots are the states in $C_i$ and each tree contains $2^{n-m}$ states. Furthermore, in each such tree, there are $2^{n-m-1}$ leaves. The \emph{last} $m$ consecutive bits in \emph{each} of these leaves is the \emph{first} $m$ consecutive bits of the corresponding root, by Equation (\ref{equ:fh}). In other words, for a tree with a state $a_0,a_1,\ldots,a_{n}$ in some $C_i$ as the root, all of its $2^{n-m-1}$ leaves have the form
	\[
	y_0,\ldots,y_{n-m-2},\overline{g(a_0,\ldots,a_{m-2})+a_{m-1}},a_0,a_1,\ldots,a_{m-1} \in \F_2^n,
	\]
	where the choices for $y_0,\ldots,y_{n-m-2}$ range over all vectors in $\F_2^{n-m-1}$.
	
	Suppose that $(\bw,\widetilde{\bw})_{i,j}$ is a PCP, \ie, there is a state $\bw=w_0,\ldots,w_{n-1}$ in $C_i$ such that $\widetilde{\bw}$ is a leaf in $G_j$ with $i \neq j$. By the above analysis, the $m$-stage state $w_{n-m},\ldots,\overline{w_{n-1}}$ is in $C_{h,j}$ while the consecutive bits $w_{n-m},\ldots,w_{n-1}$ is an $m$-stage state in $C_{h,i}$, and there is a one-to-one correspondence between this latter state and $\bw$. Thus, a PCP from $G_i$ to $G_j$ in $\mathcal{G}_{F_h}$ uniquely determines an $m$-stage companion pair shared by $C_{h,i}$ and $C_{h,j}$ in $\cG_h$.
	
	On the other hand, for $1 \leq i \neq j \leq t$, any $m$-stage companion pair  %$(\bv,\widetilde{\bv})$, with
	$\bv=v_0,\ldots,v_{m-1} \in C_{h,i}$ and $\widetilde{\bv} \in C_{h,j}$, corresponds to two PCPs in $\mathcal{G}_{F_h}$. The first PCP is $(\bw_1,\widetilde{\bw}_1)_{i,j}$. Here, $\bw_1$ is an $n$-stage state in $C_i$ with $v_0,\ldots,v_{m-1}$ as its last $m$ consecutive bits and $\widetilde{\bw}_1$ is a leaf in $G_j$. The second PCP is $(\bw_2,\widetilde{\bw}_2)_{j,i}$ such that $\bw_2$ is an $n$-stage state in $C_j$ with $v_0,\ldots,\overline{v_{m-1}}$ as its last $m$ consecutive bits and $\widetilde{\bw}_2$ is a leaf in $G_i$.
	
	We have thus proved the following result.
	
	\begin{proposition}\label{prop4}
		Let $1 \leq m<n$ be integers. If we view each $m$-stage companion pair $(\bv,\widetilde{\bv})$ between distinct cycles $C_{h,i}$ and $C_{h,j}$ in $\mathcal{G}_h$ as two distinct pairs $(\bv,\widetilde{\bv})_{i,j}$ and $(\widetilde{\bv},\bv)_{j,i}$, then there is a one-to-two correspondence between the set of all $m$-stage companion pairs in $\cG_h$ and the set of all $n$-stage PCPs in $\cG_{F_h}$. 
	\end{proposition}
	
	Let $\mathcal{A}_h$ be the adjacency graph of $h(x_0,\ldots,x_{m-1})$ as defined in the CJM. Recall that $\mathcal{A}_h$ is, therefore, a simple undirected multigraph. Its vertices are the cycles $C_{h,1},\ldots,C_{h,t}$ and the edges represent companion pairs. Proposition~\ref{prop4} tells us that, if we replace each edge in $\mathcal{A}_h$ by bidirectional edges and each vertex $C_{h,i}$ by the corresponding $G_i$, then we obtain the PAG $\mathbb{G}_{F_h}$. The number $M$ of distinct spanning trees in $\mathcal{A}_h$ can be computed by the well-known BEST Theorem, \eg, as stated in~\cite[Theorem 1]{Chang2019}. Because the spanning trees in $\mathbb{G}_{F_h}$ are directional, by taking distinct roots, the graph has $t \times M$ distinct rooted spanning trees in total.
	
	\begin{example}\label{ex:workout}
		Given $n=5$ and $m=4$, we have
		\[
		F_h(x_0,x_1,x_2,x_3,x_4) = x_1+x_2+x_3+x_4 \mbox{ when }
		h(x_0,x_1,x_2,x_3)= x_0 + x_1 + x_2 + x_3.
		\]
		Figure~\ref{fig:M1} presents $\cG_{F_h}$. The cycles in $\cG_h$, whose states are $4$-state, are 
		\[
		C_{h,1}=(0), \quad C_{h,2}=(00011), \quad C_{h,3}=(00101), \quad C_{h,4}=(01111).
		\]
		Figure~\ref{fig:two} (a) gives the compressed (undirected) adjacency graph $\mathcal{A}_h$. By computing the cofactor of any entry in its derived matrix
		\[
		\begin{pmatrix*}[r]
			1 & -1 & 0 & 0~\\
			-1 & 5 & -2 & -2 ~\\
			0 & -2 & 3 & -1 ~\\
			0 & -2 & -1 & 3~
		\end{pmatrix*},
		\]
		we know that $\mathcal{A}_h$ has $8$ distinct spanning trees. The compressed PAG $\mathbb{G}_{F_h}$, where multiple directed edges (if any) are compressed into one, is in Figure~\ref{fig:two} (b). It has $12$ distinct types of rooted spanning trees, grouped based on the respective roots. The PCPs, for any $(i,j)$ with $1 \leq i \neq j \leq 4$, can be easily determined from $\cG_{F_h}$. Table~\ref{table:PCP} provides the list for ease of reference.

		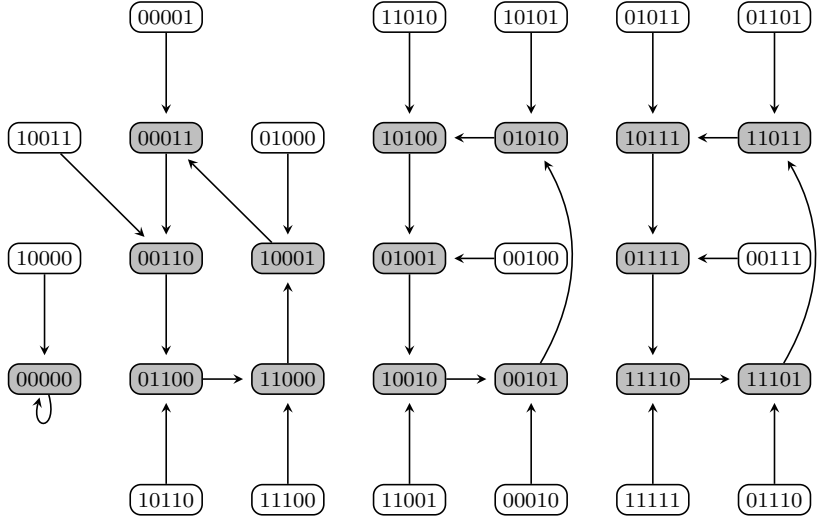
\begin{figure}[h!]
			\begin{center}
				\begin{tikzpicture}
					[
					> = stealth,
					shorten > = 3pt,
					auto,
					node distance = 1.6cm,
					semithick
					]
					
					\tikzstyle{every state}=
					\node[rectangle,fill=white,draw,rounded corners,minimum size = 4mm]
					%%% G1	
					\node[state] (1) {$10000$};
					\node[state,fill=lightgray] (2) [below of=1] {$00000$};
					
					%%% G2
					\node[state] (3) [above of=1] {$10011$};
					\node[state,fill=lightgray] (4) [right of=3] {$00011$};
					\node[state,fill=lightgray] (5) [below of=4] {$00110$};
					\node[state,fill=lightgray] (6) [below of=5] {$01100$};
					\node[state,fill=lightgray] (7) [right of=6] {$11000$};
					\node[state,fill=lightgray] (8) [above of=7] {$10001$};
					\node[state] (9) [above of=4] {$00001$};
					\node[state] (10) [below of=6] {$10110$};
					\node[state] (11) [below of=7] {$11100$};
					\node[state] (12) [above of=8] {$01000$};
					
					%%% G3
					\node[state,fill=lightgray] (14) [right of=12] {$10100$};
					\node[state,fill=lightgray] (15) [below of=14] {$01001$};
					\node[state,fill=lightgray] (16) [below of=15] {$10010$};
					\node[state] (17) [below of=16] {$11001$};
					\node[state] (18) [right of=17] {$00010$};
					\node[state,fill=lightgray] (19) [above of=18] {$00101$};
					\node[state] (20) [above of=19] {$00100$};
					\node[state,fill=lightgray] (21) [above of=20] {$01010$};
					\node[state] (22) [above of=21] {$10101$};
					\node[state] (13) [above of=14] {$11010$};
					
					%%%% G4
					\node[state,fill=lightgray] (24) [right of=21] {$10111$};
					\node[state,fill=lightgray] (25) [below of=24] {$01111$};
					\node[state,fill=lightgray] (26) [below of=25] {$11110$};
					\node[state] (27) [below of=26] {$11111$};
					\node[state] (28) [right of=27] {$01110$};
					\node[state,fill=lightgray] (29) [above of=28] {$11101$};
					\node[state] (30) [above of=29] {$00111$};
					\node[state,fill=lightgray] (31) [above of=30] {$11011$};
					\node[state] (32) [above of=31] {$01101$};
					\node[state] (23) [above of=24] {$01011$};
					
					%% G1	
					\path[->] (1) edge (2);
					\path[->] (2) edge[loop below] (2);
					
					%%% G2
					\path[->] (3) edge (5);
					\path[->] (4) edge (5);
					\path[->] (5) edge (6);
					\path[->] (10) edge (6);
					\path[->] (6) edge (7);
					\path[->] (11) edge (7);	
					\path[->] (7) edge (8);
					\path[->] (12) edge (8);
					\path[->] (8) edge (4);
					\path[->] (9) edge (4);
					
					%%%% G3
					\path[->] (13) edge (14);
					\path[->] (14) edge (15);
					\path[->] (15) edge (16);	
					\path[->] (17) edge (16);
					\path[->] (16) edge (19);	
					\path[->] (18) edge (19);
					\path[->] (19) edge[bend right] (21);
					\path[->] (20) edge (15);	
					\path[->] (21) edge (14);
					\path[->] (22) edge (21);
					
					%%%% G4
					\path[->] (23) edge (24);
					\path[->] (24) edge (25);
					\path[->] (25) edge (26);	
					\path[->] (27) edge (26);
					\path[->] (26) edge (29);	
					\path[->] (28) edge (29);
					\path[->] (29) edge[bend right] (31);
					\path[->] (30) edge (25);	
					\path[->] (31) edge (24);
					\path[->] (32) edge (31);
				\end{tikzpicture}
			\end{center}
			\caption{The state graph $\cG_{F_h}$ of $f(x_0,x_1,x_2,x_3,x_4)=x_1+x_2+x_3+x_4$ with the components $G_1$, $G_2$, $G_3$, and $G_4$ ordered left to right, \ie, $C_1=(0)$, $C_2=(10001)$, $C_3=(01010)$, and $C_4=(11011)$.} \label{fig:M1}
		\end{figure}
		
		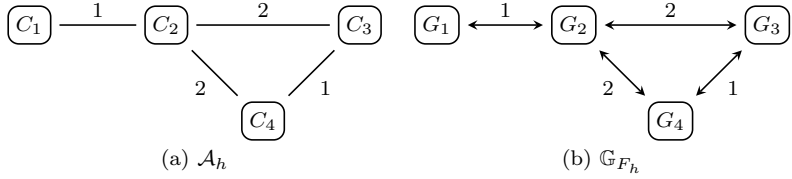
\begin{figure}
			\centering
			\begin{tabular}{cc}
				
				\begin{tikzpicture}
					[> = stealth,
					shorten > = 3pt,
					shorten < = 3pt,
					auto,
					node distance = 1.8cm,
					semithick
					]
					
					\tikzstyle{every state}=
					\node[rectangle,fill=white,draw,rounded corners,minimum size = 5mm]
					
					\node[state] (1) {$C_1$};
					\node[state] (2) [right of=1] {$C_2$};
					\node[state] (4) [below right of=2] {$C_4$};
					\node[state] (3) [above right of=4] {$C_3$};
					\draw[-] (1)--(2) node[midway,above] {$1$};
					\draw[-] (2)--(3) node[midway,above] {$2$};
					\draw[-] (4)--(2) node[midway] {$2$};
					\draw[-] (3)--(4) node[midway] {$1$};
				\end{tikzpicture}% pic 1
				& 
				
				\begin{tikzpicture}
					[> = stealth,
					shorten > = 3pt,
					shorten < = 3pt,
					auto,
					node distance = 1.8cm,
					semithick
					]
					
					\tikzstyle{every state}=
					\node[rectangle,fill=white,draw,rounded corners,minimum size = 5mm]
					
					\node[state] (1) {$G_1$};
					\node[state] (2) [right of=1] {$G_2$};
					\node[state] (4) [below right of=2] {$G_4$};
					\node[state] (3) [above right of=4] {$G_3$};
					\draw[<->] (1)--(2) node[midway,above] {$1$};
					\draw[<->] (2)--(3) node[midway,above] {$2$};
					\draw[<->] (4)--(2) node[midway] {$2$};
					\draw[<->] (3)--(4) node[midway] {$1$};
				\end{tikzpicture}
				
				\\
				
				(a) $\mathcal{A}_h$ 
				& 
				(b) $\mathbb{G}_{F_h}$
			\end{tabular}
			\caption{Relevant graphs in Example~\ref{ex:workout}. For the compressed adjacency graph $\mathcal{A}_h$, all edges between adjacent vertices are merged into one. The label gives the number of edges, \ie, the number of shared companion pairs. For the compressed preference adjacency graph $\mathbb{G}_{F_h}$, all directed edges between adjacent vertices are merged into one. The label gives the number of PCP in \emph{each} direction.}
			\label{fig:two}
		\end{figure}
		
		\begin{table}[h!]
			\caption{List of Preference Companion Pairs for $F_h$ in Example~\ref{ex:workout}}
			\label{table:PCP}
			\renewcommand{\arraystretch}{1.3}
			\centering
			\begin{tabular}{cl|cl}
				\hline
				$(i,j)$ & PCPs $\{(\bw,\widetilde{\bw})_{i,j}\}$ &
				$(i,j)$ & PCPs $\{(\bw,\widetilde{\bw})_{i,j}\}$\\
				\hline
				$(1,2)$ & $\{(00000,00001)\}$ & $(2,1)$ & $\{(10001,10000)\}$ \\
				
				$(2,3)$ & $\{(00011,00010),(11000,11001)\}$ &
				
				$(3,2)$ & $\{(01001,01000),(10010,10011)\}$ \\
				
				$(2,4)$ & $\{(00110,00111),(01100,01101)\}$ &
				
				$(4,2)$ & $\{(10111,10110),(11101,11100)\}$ \\
				
				$(3,4)$ & $\{(01010,01011)\}$ &
				
				$(4,3)$ & $\{(11011,11010)\}$ \\
				\hline
			\end{tabular}
		\end{table}
		
		The number of de Bruijn sequences that Algorithm~\ref{algo:new} outputs can now be easily determined. From Figure~\ref{fig:RST} (a) and since $C_1$ is a loop $(0)$, the algorithm produces $8$ de Bruijn sequences. There are five states each in $C_k$ for $k \in \{2,3,4\}$. Hence, reading from Figures~\ref{fig:RST} (b), (c), and (d), the algorithm yields $120$ more sequences. Alternatively, note that the cofactor of any entry in the derived matrix of $\mathbb{G}_{F_h}$ is $8$. The product of this cofactor and the number of $5$-stage states in the union of cycles, which is $16$, is $128$. These are the $128$ choices for the input $(f,\bu)$ in the GJPO Algorithm. Since the number of distinct rooted spanning trees is $8 \cdot 4 = 32$, there are at least $32$ inequivalent de Bruijn sequences among the $128$. In fact, there are $70$ sequences that appear once each. There are $23$ sequences that appear twice each. There is one sequence that appears three, four, and five times, respectively. These are 
		\begin{align*}
			&(00000100~10111110~10100011~01100111),\\ &(00000100~01110101~00110110~01011111),\\ &(00000101~11000111~11010100~11011001).
		\end{align*}
		Thus, Algorithm~\ref{algo:new} produces $96$ inequivalent de Bruijn sequences. \QEDB
		
		\begin{figure*}[h!]
			\centering
			\begin{tabular}{cc}
				\begin{tikzpicture}
					[
					> = stealth,
					shorten > = 3pt,
					auto,
					node distance = 1.5cm,
					semithick
					]
					
					\tikzstyle{every state}=
					\node[rectangle,fill=white,draw,rounded corners,minimum size = 5mm]
					
					\node[state,fill=blue!20] (1) {$G_1$};
					\node[state] (2) [right of=1] {$G_2$};
					\node[state] (3) [right of=2] {$G_3$};
					\node[state] (4) [right of=3] {$G_4$};
					\path[->] (4) edge node [above] {$1$} (3);
					\path[->] (3) edge node [above] {$2$} (2);
					\path[->] (2) edge node [above] {$1$} (1);	
					
				\end{tikzpicture} &
				
				\begin{tikzpicture}
					[
					> = stealth,
					shorten > = 3pt,
					auto,
					node distance = 1.5cm,
					semithick
					]
					
					\tikzstyle{every state}=
					\node[rectangle,fill=white,draw,rounded corners,minimum size = 5mm]
					
					\node[state] (1) {$G_1$};
					\node[state,fill=blue!20] (2) [right of=1] {$G_2$};
					\node[state] (3) [right of=2] {$G_3$};
					\node[state] (4) [right of=3] {$G_4$};
					\path[->] (4) edge node [above] {$1$} (3);
					\path[->] (3) edge node [above] {$2$} (2);
					\path[->] (1) edge node [above] {$1$} (2);	
				\end{tikzpicture}\\
				
				%%%%
				\begin{tikzpicture}
					[
					> = stealth,
					shorten > = 3pt,
					auto,
					node distance = 1.5cm,
					semithick
					]
					
					\tikzstyle{every state}=
					\node[rectangle,fill=white,draw,rounded corners,minimum size = 5mm]
					
					\node[state,fill=blue!20] (1) {$G_1$};
					\node[state] (2) [right of=1] {$G_2$};
					\node[state] (4) [right of=2] {$G_4$};
					\node[state] (3) [right of=3] {$G_3$};
					\path[->] (3) edge node [above] {$1$} (4);
					\path[->] (4) edge node [above] {$2$} (2);
					\path[->] (2) edge node [above] {$1$} (1);	
					
				\end{tikzpicture} &
				
				\begin{tikzpicture}
					[
					> = stealth,
					shorten > = 3pt,
					auto,
					node distance = 1.5cm,
					semithick
					]
					
					\tikzstyle{every state}=
					\node[rectangle,fill=white,draw,rounded corners,minimum size = 5mm]
					
					\node[state] (1) {$G_1$};
					\node[state,fill=blue!20] (2) [right of=1] {$G_2$};
					\node[state] (4) [right of=2] {$G_4$};
					\node[state] (3) [right of=4] {$G_3$};
					\path[->] (3) edge node [above] {$1$} (4);
					\path[->] (4) edge node [above] {$2$} (2);
					\path[->] (1) edge node [above] {$1$} (2);	
				\end{tikzpicture}\\
				
				%%%%%
				\begin{tikzpicture}
					[
					> = stealth,
					shorten > = 2pt,
					auto,
					node distance = 1.5cm,
					semithick
					]
					
					\tikzstyle{every state}=
					\node[rectangle,fill=white,draw,rounded corners,minimum size = 5mm]
					
					\node[state,fill=blue!20] (1) {$G_1$};
					\node[state] (2) [right of=1] {$G_2$};
					\node[state] (3) [right of=2] {$G_3$};
					\node[state] (4) [right of=3] {$G_4$};
					\path[->] (4) edge[bend right] node [above] {$2$} (2);
					\path[->] (3) edge node [above] {$2$} (2);
					\path[->] (2) edge node [above] {$1$} (1);	
					
				\end{tikzpicture} &
				
				\begin{tikzpicture}
					[
					> = stealth,
					shorten > = 2pt,
					auto,
					node distance = 1.5cm,
					semithick
					]
					
					\tikzstyle{every state}=
					\node[rectangle,fill=white,draw,rounded corners,minimum size = 5mm]
					
					\node[state] (1) {$G_1$};
					\node[state,fill=blue!20] (2) [right of=1] {$G_2$};
					\node[state] (4) [right of=2] {$G_4$};
					\node[state] (3) [right of=4] {$G_3$};
					\path[->] (4) edge node [above] {$2$} (2);
					\path[->] (3) edge[bend right] node [above] {$2$} (2);
					\path[->] (1) edge node [above] {$1$} (2);	
				\end{tikzpicture}
				\bigskip
				\\

				(a) {Spanning Trees with Root $G_1$} & (b) {Spanning Trees with Root $G_2$}
				\bigskip
				\\
				
				%%%%
				
				\begin{tikzpicture}
					[
					> = stealth,
					shorten > = 3pt,
					auto,
					node distance = 1.5cm,
					semithick
					]
					
					\tikzstyle{every state}=
					\node[rectangle,fill=white,draw,rounded corners,minimum size = 5mm]
					
					\node[state] (1) {$G_1$};
					\node[state] (2) [right of=1] {$G_2$};
					\node[state,fill=blue!20] (3) [right of=2] {$G_3$};
					\node[state] (4) [right of=3] {$G_4$};
					\path[->] (4) edge node [above] {$1$} (3);
					\path[->] (2) edge node [above] {$2$} (3);
					\path[->] (1) edge node [above] {$1$} (2);	
					
				\end{tikzpicture} &
				
				\begin{tikzpicture}
					[
					> = stealth,
					shorten > = 3pt,
					auto,
					node distance = 1.5cm,
					semithick
					]
					
					\tikzstyle{every state}=
					\node[rectangle,fill=white,draw,rounded corners,minimum size = 5mm]
					
					\node[state] (1) {$G_1$};
					\node[state] (2) [right of=1] {$G_2$};
					\node[state] (3) [right of=2] {$G_3$};
					\node[state,fill=blue!20] (4) [right of=3] {$G_4$};
					\path[->] (3) edge node [above] {$1$} (4);
					\path[->] (2) edge node [above] {$2$} (3);
					\path[->] (1) edge node [above] {$1$} (2);	
				\end{tikzpicture}\\
				
				%%%%
				\begin{tikzpicture}
					[
					> = stealth,
					shorten > = 3pt,
					auto,
					node distance = 1.5cm,
					semithick
					]
					
					\tikzstyle{every state}=
					\node[rectangle,fill=white,draw,rounded corners,minimum size = 5mm]
					
					\node[state] (1) {$G_1$};
					\node[state] (2) [right of=1] {$G_2$};
					\node[state] (4) [right of=2] {$G_4$};
					\node[state,fill=blue!20] (3) [right of=4] {$G_3$};
					\path[->] (1) edge node [above] {$1$} (2);
					\path[->] (2) edge node [above] {$2$} (4);
					\path[->] (4) edge node [above] {$1$} (3);	
					
				\end{tikzpicture} &
				
				\begin{tikzpicture}
					[
					> = stealth,
					shorten > = 3pt,
					auto,
					node distance = 1.5cm,
					semithick
					]
					
					\tikzstyle{every state}=
					\node[rectangle,fill=white,draw,rounded corners,minimum size = 5mm]
					
					\node[state] (1) {$G_1$};
					\node[state] (2) [right of=1] {$G_2$};
					\node[state,fill=blue!20] (4) [right of=2] {$G_4$};
					\node[state] (3) [right of=4] {$G_3$};
					\path[->] (2) edge node [above] {$2$} (4);
					\path[->] (3) edge node [above] {$1$} (4);
					\path[->] (1) edge node [above] {$1$} (2);	
				\end{tikzpicture}\\
				
				%%%%%
				\begin{tikzpicture}
					[
					> = stealth,
					shorten > = 2pt,
					auto,
					node distance = 1.5cm,
					semithick
					]
					
					\tikzstyle{every state}=
					\node[rectangle,fill=white,draw,rounded corners,minimum size = 5mm]
					
					\node[state] (1) {$G_1$};
					\node[state] (2) [right of=1] {$G_2$};
					\node[state,fill=blue!20] (3) [right of=2] {$G_3$};
					\node[state] (4) [right of=3] {$G_4$};
					\path[->] (4) edge[bend right] node [above] {$2$} (2);
					\path[->] (2) edge node [above] {$2$} (3);
					\path[->] (1) edge node [above] {$1$} (2);	
					
				\end{tikzpicture} &
				
				\begin{tikzpicture}
					[
					> = stealth,
					shorten > = 2pt,
					auto,
					node distance = 1.5cm,
					semithick
					]
					
					\tikzstyle{every state}=
					\node[rectangle,fill=white,draw,rounded corners,minimum size = 5mm]
					
					\node[state] (1) {$G_1$};
					\node[state] (2) [right of=1] {$G_2$};
					\node[state,fill=blue!20] (4) [right of=2] {$G_4$};
					\node[state] (3) [right of=4] {$G_3$};
					\path[->] (2) edge node [above] {$2$} (4);
					\path[->] (3) edge[bend right] node [above] {$2$} (2);
					\path[->] (1) edge node [above] {$1$} (2);	
				\end{tikzpicture}
				\bigskip
				\\
				
				(c) {Spanning Trees with Root $G_3$} & (d) {Spanning Trees with Root $G_4$}\\
				
			\end{tabular}
			\caption{List of $12$ distinct types of rooted spanning trees in $\mathbb{G}_{F_h}$ with multiple edges compressed into one. The label above each directed edge is the number of PCPs to choose from.}
			\label{fig:RST}
		\end{figure*}
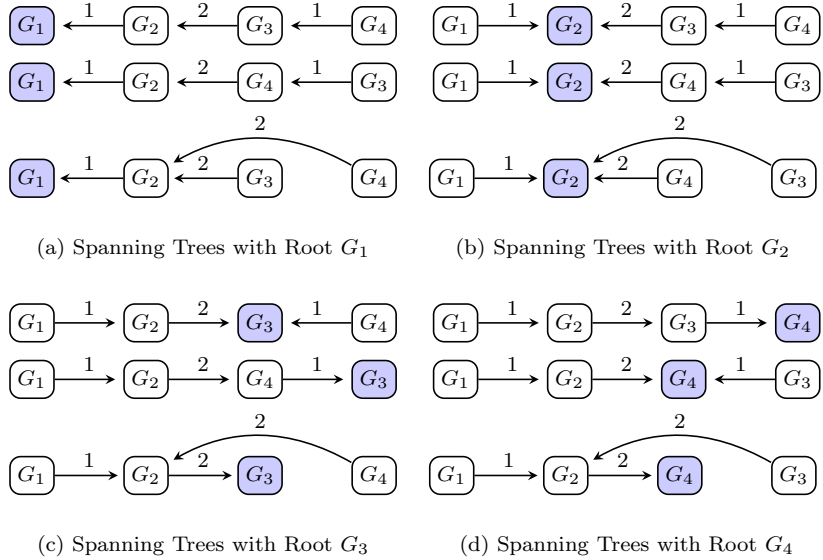
		
	\end{example}
	
	Computational evidences point to the following conjecture.
	
	\begin{conjecture}
		Given the feedback function $F_h$, defined based on a nonsingular $h$ in Equation (\ref{equ:fh}) with $n \ge m+2$, distinct input pairs of rooted spanning tree and initial state generate inequivalent de Bruijn sequences of order $n$. This holds in both the GPO and the GJPO Algorithms.
	\end{conjecture}
	
	\section{Conclusion and Future Directions}\label{sec:conclu}
	
	We recap what this work has accomplished and point out several directions to explore. In terms of the GPO Algorithm, we have characterized a set of conditions for which the algorithm is certified to produce de Bruijn sequences. Carefully constructed classes of feedback functions, as demonstrated in Section~\ref{sec:PAG}, can be easily shown to meet the conditions.
	
	The insights learned from studying the GPO Algorithm led us to a modification, which we call the {\tt Graph Joining Prefer-Opposite} (GJPO) Algorithm, that greatly enlarges the classes of feedback functions that can be used to greedily generate de Bruijn sequences. We adapt several key steps from the cycle joining method to join graph components in the state graph of any suitable function and initial state pair. The use of greedy algorithms to generate de Bruijn sequences remains of deep theoretical attraction, despite their practical drawbacks. Once some special state is visited, then we can be sure that all other states must have been visited before. This may be of independent interest and could be useful in other domains, \eg, in the designs and verification of experiments.
	
	There are several important challenges to overcome if we are to turn this novel idea of graph joining into a more practical tool. First, the task of identifying general classes of feedback functions whose respective state graphs are easy to characterize and efficient to store and manipulate may be the most important open direction. Accomplishing this can greatly reduces the complexity of the two proposed algorithms. We have discussed some examples of such classes, yet we believe that many better ones must exist and are waiting to be discovered.
	Second, instead of determining all preference companion pairs (PCPs), one can opt to seek for tools to quickly identify some of them to build a fast algorithm with low memory requirement. Third, studying the suitability of the resulting sequences for specific application domain(s) may appeal more to practitioners.
	
	\smallskip
	
	\noindent
	{\bf Comparison with the Cycle Joining Method}
	
	In the GJM, the feedback functions are {\bf singular}, \ie, the coefficient of $x_0$ is $0$. The state graph of such a function is a {\bf union of subgraphs}, not a union of cycles. In our case, the CJM does not apply since, as the name clearly indicates, the CJM joins cycles. It is not designed to join trees, for instance. Our GJM assigns the successor of a special state to visit all states of the corresponding subgraph by a greedy algorithm. Using the GJM, we can join subgraphs (not just cycles) and  generate many de Bruijn sequences, including those in~\cite{WWZ18} and the 
	{\tt Prefer-Opposite} sequence~\cite{Alh10}. In their respective algorithmic implementations, the GJM is greedy while the CJM is not. 
	
	In the CJM, there is a bijection between the set of spanning trees in the adjacency graph and the collection of all de Bruijn sequences in its output. In the GJM, the (preference) adjacency graph is a {\bf directed graph}. The GJM can generate a de Bruijn sequence only if there is a rooted spanning tree such that the initial state is in the cycle corresponding to the root. The nice thing is that one can adopt a crucial step in the CJM that identifies spanning trees to find rooted spanning trees in the GJM. We emphasize once again that, in the adjacency graph of the CJM, the vertices are cycles, while the vertices in the directed adjacency graph of the GJM are subgraphs. A recent work of Gabric \etal in~\cite{Gabric20} relies on the CJM because, there, the state graph is a union of cycles. Hence, our work is not a mere extension of their treatment.
	
	We briefly mention the connections between the two methods and well-known algorithms that deal with Eulerian cycles in graphs. The CJM can be viewed as a special case of the Hierholzer Algorithm~\cite{Hierholzer}. Theorem~\ref{thm:m} can be viewed as a special case of the Fleury Algorithm~\cite{Fleury} since there is exactly one cycle in the corresponding state graph, {\it i.e.}, the state graph is connected. The GJPO Algorithm, in contrast, is not a special case of Fleury's. A neccessary but insufficient condition for the Fleury Algorithm is that all of the edges must be in the same component. Hence, the Fleury Algorithm, unlike the GJPO one, does not work on graphs with multiple components.
	
	\section*{Acknowledgement}
	The work of Z.~Chang is supported by the National Natural Science Foundation of China under Grant 61772476. Nanyang Technological University Grant 04INS000047C230GRT01 supports the research carried out by M.~F.~Ezerman and A.~A.~Fahreza. Q.~Wang is supported by an individual grant from the Natural Sciences and Engineering Research Council of Canada.
	
	%\bibliographystyle{IEEEtran}
	%\bibliography{greedy}
	% Generated by IEEEtran.bst, version: 1.14 (2015/08/26)

\end{document}